\documentclass[11pt]{llncs}
\usepackage{epsfig,latexsym}

\usepackage{color}
\newcommand{\definef}[1]{{\color{blue} #1}} 
\newcommand{\commentout}[1]{}

\newcommand{\cH}{{\cal H}}

\newcommand{\dY}{\downarrow$$Y}

\newcommand{\cT}{{\cal T}}
\newcommand{\cE}{{\cal E}}

\newcommand{\tb}{{\sf tb}}
\newcommand{\tw}{{\sf tw}}
\newcommand{\tl}{{\sf tl}}
\newcommand{\w}{{\sf w}}

\begin{document}
\title{Collective Additive Tree Spanners of Bounded Tree-Breadth Graphs with Generalizations and Consequences
}


\author{Feodor F. Dragan \and  Muad Abu-Ata
}

\institute{Algorithmic Research Laboratory, Department of Computer
Science, \\ Kent State University, Kent, OH 44242, USA  \\
{\em \{dragan, mabuata\}@cs.kent.edu}}

\maketitle

\begin{abstract}
In this paper, we study collective additive tree spanners for
families of graphs enjoying special Robertson-Seymour's
tree-decompositions, and demonstrate 
interesting consequences of obtained results. We say that a graph
$G$ {\em admits a system of $\mu$ collective additive tree
$r$-spanners} (resp., {\em multiplicative tree $t$-spanners}) if
there is a system $\cT(G)$ of at most $\mu$ spanning trees of $G$
such that for any two vertices $x,y$ of $G$ a spanning tree $T\in
\cT(G)$ exists such that $d_T(x,y)\leq d_G(x,y)+r$ (resp.,
$d_T(x,y)\leq t\cdot d_G(x,y)$). When $\mu=1$ one gets the notion of
{\em additive tree $r$-spanner} (resp., {\em
multiplicative tree $t$-spanner}). It is known 
that if a graph $G$ has a multiplicative tree $t$-spanner, then $G$
admits a Robertson-Seymour's tree-decomposition with bags of radius
at most $\lceil{t/2}\rceil$ in $G$. We use this to demonstrate that
there is a polynomial time algorithm that, given an $n$-vertex graph
$G$ admitting a multiplicative tree $t$-spanner, constructs a system
of at most $\log_2 n$ collective additive tree $O(t\log n)$-spanners
of $G$. That is, with a slight increase in the number of trees and
in the stretch, one can ``turn" a multiplicative tree spanner into a
small set of collective additive tree spanners. We extend this
result by showing that if a graph $G$ admits a multiplicative
$t$-spanner with tree-width $k-1$, then $G$ admits a
Robertson-Seymour's tree-decomposition each bag of which can be
covered with at most $k$ disks of $G$ of radius at most
$\lceil{t/2}\rceil$ each. This is used to demonstrate that, for
every fixed $k$, there is a polynomial time algorithm that, given an
$n$-vertex graph $G$ admitting a multiplicative $t$-spanner with
tree-width $k-1$, constructs a system of at most $k(1+ \log_2 n)$
collective additive tree $O(t\log n)$-spanners of $G$.

\medskip
\noindent{\bf Keywords:} {\em graph algorithms; approximation
algorithms; tree spanner problem; collective tree spanners; spanners
of bounded tree-width;  Robertson-Seymour's tree-decomposition;
balanced separators.}
\end{abstract}

\section{Introduction}

One of the basic questions  in the design of routing schemes for
communication  networks is to construct a spanning network (a
so-called spanner) which has
 two (often conflicting) properties: it should have simple
structure and nicely approximate distances in the network. This
problem  fits in a larger framework of combinatorial and algorithmic
problems that are concerned with distances in a finite metric space
induced by a graph. An arbitrary metric space (in particular a
finite metric defined by a graph) might not have enough structure to
exploit algorithmically. A powerful technique that has been
successfully used recently in this context is to embed the given
metric space in a simpler metric space such that the distances are
approximately preserved in the embedding. New and improved
algorithms have resulted from this idea for several important
problems (see, e.g.,
\cite{Bartal96,BartalBBT97,CharikarCGGP98,ElkinEST08,GuptaKR04,LinialLR95}).

There are several ways to measure the quality of this approximation,
two of them leading to the notion of a spanner. For $t\geq 1$, a
spanning subgraph $H$ of $G=(V,E)$ is called a
\emph{(multiplicative) $t$-spanner} of $G$ if $d_H(u,v)\leq t\cdot
d_G(u,v)$ for all $u,v\in V$ \cite{Chew,PeSc,PelUll}. If $r\geq 0$
and $d_H(u,v)\leq  d_G(u,v) + r$, for all $u,v\in V$, then $H$ is
called an \emph{additive $r$-spanner} of $G$
\cite{LieShe,Pri96,Kratsch}. The parameter $t$ is called the {\em
stretch} (or {\em stretch factor}) of $H$, while the parameter $r$
is called the {\em surplus} of $H$. In what follows, we will often
omit the word ``multiplicative" when we refer to multiplicative
spanners.

Tree metrics are a very natural class of simple metric spaces since
many algorithmic problems become tractable on them. A
\emph{(multiplicative) tree $t$-spanner}  of a graph $G$ is a
spanning tree with a stretch $t$ \cite{CaiC95}, and an
\emph{additive tree $r$-spanner}  of $G$ is a spanning tree with a
surplus $r$ \cite{Pri96}. If we approximate the graph by a tree
spanner, we can solve the problem on the tree and the solution
interpret on the original graph. The {\sc tree $t$-spanner} problem
asks, given a graph $G$ and a positive number $t$, whether $G$
admits a tree $t$-spanner. Note that the problem of finding a tree
$t$-spanner of $G$ minimizing $t$ is known in literature also as the
Minimum Max-Stretch spanning Tree problem (see, e.g., \cite{EmekP04}
and literature cited therein).

Unfortunately, not many graph families admit good tree spanners.
This motivates the study of sparse spanners, i.e., spanners with a
small amount of edges. There are many applications of spanners in
various areas; especially, in distributed systems and communication
networks. In \cite{PelegU89}, close relationships were established
between the quality of spanners (in terms of stretch factor and the
number of spanner edges), and the time and communication
complexities of any synchronizer for the network based on this
spanner. Another example is the usage of tree $t$-spanners in the
analysis of arrow distributed queuing protocols
\cite{HeKu06,PeRe01}. Sparse spanners are very useful in message
routing in communication networks; in order to maintain succinct
routing tables, efficient routing schemes can use only the edges of
a sparse spanner \cite{PelegU88}. The {\sc Sparsest $t$-Spanner}
problem asks, for a given graph $G$ and a number $t$, to find a
$t$-spanner of $G$ with the smallest number of edges. We refer to
the survey paper  of Peleg \cite{Peleg02} for an overview on
spanners.

Inspired by ideas from works of Alon et al. \cite{NoKaPe}, Bartal
\cite{Bartal96,Ba-5},  Fakcharoenphol et al.
\cite{FakcharoenpholRT03}, and to extend those ideas to designing
compact and efficient routing and distance labeling schemes in
networks,  in \cite{CollTrSp},  a new notion of {\em collective tree
spanners}\footnote{Independently, Gupta et al. in \cite{GuptaKR04}
introduced a similar concept which is called {\em  tree covers}
there.} were introduced. This notion slightly {\sl weaker} than the
one of a tree spanner and slightly {\sl stronger} than the notion of
a sparse spanner. We say that a graph $G=(V,E)$ {\em admits a system
of $\mu$ collective additive tree $r$-spanners} if there is a system
$\cT(G)$ of at most $\mu$ spanning trees of $G$ such that for any
two vertices $x,y$ of $G$ a spanning tree $T\in \cT(G)$ exists such
that $d_T(x,y)\leq d_G(x,y)+r$ (a multiplicative variant of this
notion can be defined analogously). Clearly, if $G$ admits a system
of $\mu$ collective additive tree $r$-spanners, then $G$ admits an
additive $r$-spanner with at most $\mu\times (n-1)$ edges (take the
union of all those trees), and if $\mu=1$ then $G$ admits an
additive tree $r$-spanner.

Recently, in \cite{DrFoGo}, {\em spanners of bounded tree-width}
were introduced, motivated by the fact that many algorithmic
problems are tractable on graphs of bounded tree-width, and a
spanner $H$ of $G$ with small tree-width can be used to obtain an
approximate solution to a problem on $G$. In particular, efficient
and compact distance and routing labeling schemes are available for
bounded tree-width graphs (see, e.g., \cite{CollTrSpPar,GuptaKR04}
and papers cited therein), and they can be used to compute
approximate distances and route along paths that are close to
shortest in $G$. The \textsc{$k$-Tree-width $t$-spanner} problem
asks, for a given graph $G$, an integers $k$ and a positive number
$t\geq 1$, whether $G$ admits a $t$-spanner of tree-width at most
$k$. Every connected graph with $n$ vertices and at most $n-1+m$
edges is of tree-width at most $m+1$ and hence this problem is a
generalization of the {\sc Tree $t$-Spanner} and the {\sc Sparsest
$t$-Spanner} problems. Furthermore,  $t$-spanners of bounded
tree-width have much more structure to exploit algorithmically than
sparse $t$-spanners (which have a small number of edges but may lack
other nice structural properties).

\subsection{Related work}
\paragraph{Tree spanners.} Substantial work has been done on the {\sc tree $t$-spanner} problem
on unweighted graphs. Cai and Corneil  \cite{CaiC95} have shown
that, for a given graph $G$, the problem to decide whether $G$ has a
tree $t$-spanner is NP-complete for any fixed $t\geq 4$ and is
linear time solvable for $t=1,2$ (the status of the case $t=3$ is
open for general graphs)\footnote{When $G$ is an unweighted graph,
$t$ can be assumed to be an integer.}. The NP-completeness result
was further strengthened in \cite{BrDrLeLe} and \cite{BrDrLeLeUe},
where Branst\"adt et al. showed that the problem remains NP-complete
even for the class of chordal graphs (i.e., for graphs where each
induced cycle has length 3) and every fixed $t\geq 4$, and for the
class of chordal bipartite graphs (i.e., for bipartite graphs where
each induced cycle has length 4) and every fixed $t\geq 5$.

The {\sc tree $t$-spanner} problem on planar graphs was studied in
\cite{DrFoGo,FeketeK01}. 
In \cite{FeketeK01}, Fekete  and Kremer  proved that the {\sc tree
$t$-spanner} problem  on planar graphs is NP-complete (when  $t$ is
part of the input) and polynomial time solvable for $t=3$.
For fixed $t\geq 4$,  the complexity of the {\sc tree $t$-spanner}
problem on arbitrary planar graphs was left as an open problem in
\cite{FeketeK01}. This open problem was recently resolved in
\cite{DrFoGo} by Dragan et al., where it was shown that the {\sc
tree $t$-spanner} problem  is linear time solvable for every fixed
constant $t$ on the class of apex-minor-free graphs which includes
all planar graphs and all graphs of bounded genus.
Note also that a number of particular graph classes (like interval
graphs, permutation graphs, asteroidal-triple--free graphs, strongly
chordal graphs, dually chordal graphs, and others) admit additive
tree $r$-spanners for small values of $r$ (we refer reader to
\cite{BrChDr,BrDrLeLe,BrDrLeLeUe,CaiC95,FeketeK01,LiWu08,Peleg02,PeRe01,Pri96,Kratsch}
and papers cited therein).

The first  $O(\log n)$-approximation algorithm for the minimum value
of $t$ for the {\sc tree $t$-spanner} problem was developed by  Emek
and Peleg in \cite{EmekP04} (where $n$ is the number of vertices in a graph). 
Recently, another logarithmic approximation algorithm for the
problem  was proposed in \cite{APPROX2011} (we elaborate more on
this in Subsection \ref{our-res}). Emek and Peleg also established
in \cite{EmekP04} that unless P = NP, the problem cannot be
approximated additively by any $o(n)$ term. Hardness of
approximation is established also in \cite{LiWu08}, where it was
shown that approximating  the minimum value of $t$ for the {\sc tree
$t$-spanner} problem within factor better than 2 is NP-hard (see
also \cite{PeRe01} for an earlier result).


\paragraph{Sparse spanners.}
Sparse $t$-spanners were introduced by Peleg, Sch{\"a}ffer and
Ullman in \cite{PeSc,PelegU89} and since that time were studied
extensively. It was shown by Peleg and  Sch{\"a}ffer in \cite{PeSc}
that the problem of deciding whether a graph $G$ has a $t$-spanner
with at most $m$ edges is NP-complete. Later, Kortsarz
\cite{Kortsarz01} showed that for every $t\geq 2$ there is a
constant $c<1$ such that it is NP-hard to approximate the sparsest
$t$-spanner within the ratio $c\cdot\log n$, where $n$ is the number
of vertices in the graph. On the other hand, the problem admits a
$O(\log n)$-ratio approximation for $t=2$
\cite{KortsarzP94,Kortsarz01} and a $O(n^{2/(t+1)})$-ratio
approximation for $t>2$ \cite{ElPe}. For some other
inapproximability and approximability results for the {\sc Sparsest
$t$-Spanner} problem on general graphs we refer the reader to
\cite{BBMRY,BGJRW-09,DKR-ICALP2012,DK-STOC2011,ElkinP00b,ElkinP00a,ElPe,ThorupZ05}
and papers cited therein. It is interesting to note also that any
(even weighted) $n$-vertex graph admits an $O(2k-1)$-spanner with at
most $O(n^{1+1/k})$ edges for any $k\geq 1$, and such a spanner can
be constructed in polynomial time
\cite{mspanner1,mspanner2,ThorupZ05}.

On planar graphs the {\sc Sparsest $t$-Spanner} problem was studied
as well. Brandes and Handke have shown that the decision version of
the problem remains NP-complete on planar graphs for every fixed
$t\geq 5$ (the case $2\leq t\leq 4$ is open) \cite{BrandesH98}.
Duckworth, Wormald, and  Zito  \cite{DuckworthWZ03} have shown that
the problem of finding a sparsest $2$-spanner of a $4$-connected
planar triangulation admits a polynomial time approximation scheme
(PTAS). Dragan et al.~\cite{DrFoGoPTAS} proved that the {\sc
Sparsest $t$-Spanner} problem admits PTAS for graph classes of
bounded local tree-width (and therefore for planar and bounded genus
graphs).

Sparse additive spanners were considered in
\cite{aspanner1,aspanner2,aspanner3,LieShe,Woodruff-ICALP10}. It is
known that every $n$-vertex graph admits an additive 2-spanner with
at most $\Theta(n^{3/2})$ edges \cite{aspanner2,aspanner3}, an
additive 6-spanner with at most $O(n^{4/3})$ edges \cite{aspanner1},
and an additive $O(n^{(1-1/k)/2})$-spanner with at most
$O(n^{1+1/k})$ edges for any $k\geq 1$ \cite{aspanner1}. All those
spanners can be constructed in polynomial time. We refer the reader
to paper \cite{Woodruff-ICALP10} for a good summary of the state of
the art results on the sparsest additive spanner problem in general
graphs.

\paragraph{Collective tree spanners.}
The problem of finding ``small'' systems of collective additive tree
$r$-spanners for small values of $r$ was examined on special classes
of graphs in \cite{CDKY,CoDrYa,CollTrSpPar,CollTrSp,UDGmy}. For
example, in \cite{CDKY,CollTrSp}, sharp results were obtained for
unweighted chordal graphs and $c$-chordal graphs (i.e., the graphs
where each induced cycle has length at most $c$): every $c$-chordal
graph admits a system of at most $\log_{2} n$ collective additive
tree $(2\lfloor c/2\rfloor)$--spanners, constructible in polynomial
time; no system of constant number of collective additive tree
$r$-spanners can exist for chordal  graphs (i.e., when $c=3$) and
$r\leq 3$, and no system of constant number of collective additive
tree $r$-spanners can exist for outerplanar graphs for any constant
$r$.

Only papers \cite{CollTrSpPar,GuptaKR04,UDGmy} have investigated
collective (multiplicative or additive) tree spanners in {\sl
weighted graphs}. It was shown that any weighted $n$-vertex planar
graph admits a system of $O(\sqrt{n})$ collective multiplicative
tree 1-spanners (equivalently, additive tree 0-spanners)
\cite{CollTrSpPar,GuptaKR04} and a system of at most $2\log_{3/2} n$
collective multiplicative tree $3$--spanners \cite{GuptaKR04}.
Furthermore, any weighted graph with genus at most $g$ admits a
system of $O(\sqrt{gn})$ collective additive tree $0$--spanners
\cite{CollTrSpPar,GuptaKR04}, any weighted graph with tree-width at
most $k-1$ admits a system of at most $k \log_{2} n$ collective
additive tree $0$--spanners \cite{CollTrSpPar,GuptaKR04}, any
weighted graph $G$ with clique-width at most $k$ admits a system of
at most  $k \log_{3/2} n$ collective additive tree $(2\w)$--spanners
\cite{CollTrSpPar}, any weighted $c$-chordal graph $G$ admits a
system of $\log_{2} n$ collective additive tree $(2\lfloor
c/2\rfloor\w)$--spanners \cite{CollTrSpPar} (where $\w$
denotes the maximum edge weight in $G$). 

Collective tree spanners of   Unit Disk Graphs (UDGs) (which often
model wireless ad hoc networks) were investigated in \cite{UDGmy}.
It was shown that every $n$-vertex UDG $G$ admits a system ${\cal
T}(G)$ of at most $2 \log_{\frac{3}{2}} n+2$ spanning trees of $G$
such that, for any two vertices $x$ and $y$ of $G$, there exists a
tree $T$ in ${\cal T}(G)$ with $d_T(x,y)\leq 3\cdot d_G(x,y) +12$.
That is, the distances in any UDG  can be approximately represented
by the distances in at most $2 \log_{\frac{3}{2}} n+2$ of its
spanning trees. Based on this result a new {\em  compact and low
delay routing labeling scheme} was proposed for Unit Disk Graphs.

\paragraph{Spanners with bounded tree-width.}
The \textsc{$k$-Tree-width $t$-spanner} problem was considered in
\cite{DrFoGo} and \cite{boud-deg}. It was shown that the problem is
linear time solvable for every fixed constants $t$ and $k$ on the
class of apex-minor-free graphs \cite{DrFoGo}, which includes all
planar graphs and all graphs of bounded genus,  and on the graphs
with bounded degree \cite{boud-deg}.

\subsection{\bf Our results and their place in the context of the previous
results.} \label{our-res}
This paper was inspired by few recent results from
\cite{DoDrGaYa,APPROX2011,ElPe,EmekP04}.
Elkin and Peleg in \cite{ElPe}, among other results, described a
polynomial time algorithm that, given an $n$-vertex  graph $G$
admitting a tree $t$-spanner, constructs a $t$-spanner of $G$ with
$O(n\log n)$ edges.
Emek and Peleg in \cite{EmekP04} presented the first $O(\log
n)$-approximation algorithm for the minimum value of $t$ for the
{\sc tree $t$-spanner} problem. They described a polynomial time
algorithm that, given an $n$-vertex  graph $G$ admitting a tree
$t$-spanner, constructs a tree $O(t \log n)$-spanner of $G$. Later,
a simpler and faster $O(\log n)$-approximation algorithm for the
problem was given by Dragan and K\"ohler \cite{APPROX2011}. Their
result uses a new necessary condition for a graph to have a tree
$t$-spanner: if a graph $G$ has a tree $t$-spanner, then $G$ admits
a Robertson-Seymour's tree-decomposition with bags of radius at most
$\lceil{t/2}\rceil$ in $G$.

To describe the results of \cite{DoDrGaYa} and to elaborate more on
the Dragan-K\"ohler's approach, we need  to recall definitions of a
few graph parameters. They all are based on the notion of
tree-decomposition introduced by Robertson and Seymour in their work
on graph minors \cite{RobSey86}.

A \emph{tree-decomposition} of a graph $G=(V,E)$ is a  pair
$(\{X_i|i\in I\},T=(I,F))$ where $\{X_i|i\in I\}$ is a collection of
subsets of $V$, called {\em bags}, and $T$ is a tree. The nodes of
$T$ are the bags $\{X_i|i\in I\}$ satisfying the following three
conditions:
 \begin{enumerate}\vspace*{-2mm}
   \item $\bigcup_{i\in I}X_i=V$;
   \item for each edge $uv\in E$, there is a bag $X_i$ such that $u,v \in
   X_i$;
   \item for all $i,j,k \in I$, if $j$ is on the path from $i$ to $k$ in $T$, then $X_i \bigcap X_k\subseteq X_j$.
   Equivalently, this condition could be stated as follows: for all vertices $v\in V$, the set of bags $\{i\in I| v\in X_i\}$ induces a connected subtree $T_v$ of $T$.
 \end{enumerate}\vspace*{-2mm}
For simplicity we denote a tree-decomposition $\left(\{X_i|i\in
I\},T=(I,F)\right)$ of a graph $G$ by $T(G)$.

Tree-decompositions were used to define several graph parameters to
measure how close a given graph is to some known graph class (e.g.,
to trees or to  chordal graphs) where many algorithmic problems
could be solved efficiently.
The \emph{width} of a tree-decomposition $T(G)=(\{X_i|i\in
I\},T=(I,F))$ is $max_{i\in I}|X_i|-1$. The {\em tree-width} of a
graph $G$, denoted by $\tw(G)$, is the minimum width, over all
tree-decompositions $T(G)$ of $G$ \cite{RobSey86}. The trees are
exactly the graphs with tree-width 1. The {\em length} of a
tree-decomposition $T(G)$ of a graph $G$ is $\lambda:=\max_{i\in
I}\max_{u,v\in X_i}d_G(u,v)$ (i.e., each bag $X_i$ has diameter at
most $\lambda$ in $G$). The {\em tree-length} of $G$, denoted by
$\tl(G)$, is the minimum of the length, over all tree-decompositions
of $G$ \cite{DoGa2007}. The chordal graphs  are exactly the graphs
with tree-length 1. Note that these two graph parameters are not
related to each other. For instance, a clique
on $n$ vertices has tree-length 1 and tree-width $n-1$,
whereas a cycle on $3n$ vertices has tree-width 2 and tree-length
$n$. 
%
%
In \cite{APPROX2011}, yet another graph parameter was introduced,
which is very similar to the notion of tree-length and, as it turns
out,  is related to the {\sc tree $t$-spanner} problem. The {\em
breadth} of a tree-decomposition $T(G)$ of a graph $G$ is the
minimum integer $r$ such that for every $i\in I$ there is a vertex
$v_i\in V(G)$ with $X_i\subseteq D_r(v_i,G)$ (i.e., each bag $X_i$
can be covered by a disk $D_r(v_i,G):=\{u\in V(G)| d_G(u,v_i)\leq r
\}$ of radius at most $r$ in $G$). Note that vertex $v_i$ does not
need to belong to $X_i$. The {\em tree-breadth} of $G$, denoted by
$\tb(G)$, is the minimum of the breadth, over all
tree-decompositions of $G$. Evidently, for any graph $G$, $1\leq
\tb(G)\leq \tl(G)\leq 2 \tb(G)$ holds. Hence, if one parameter is
bounded by a constant for a graph $G$ then the other parameter is
bounded for $G$ as well.

We say that a family of graphs ${\cal G}$ is {\em of bounded
tree-breadth} ({\em of bounded tree-width, of bounded tree-length})
if there is a constant $c$ such that for each graph $G$ from ${\cal
G}$, $\tb(G)\leq c$ (resp., $\tw(G)\leq c$, $\tl(G)\leq c$).

\medskip

It was shown in \cite{APPROX2011} that if a graph $G$ admits a tree
$t$-spanner then its tree-breadth is at most $\lceil{t/2}\rceil$ and
its tree-length is at most $t$. Furthermore, any graph $G$ with
tree-breadth $\tb(G)\leq \rho$ admits a tree $(2\rho\lfloor\log_2
n\rfloor)$-spanner that can be constructed in polynomial time. Thus,
these two results gave a new $\log_2 n$-approximation algorithm for
the {\sc tree $t$-spanner} problem on general (unweighted) graphs
(see \cite{APPROX2011} for details). The algorithm of
\cite{APPROX2011} is conceptually simpler than the previous $O(\log
n)$-approxima\-tion algorithm proposed  for the problem by Emek and
Peleg \cite{EmekP04}.

Dourisboure et al. in \cite{DoDrGaYa} concerned with the
construction of additive spanners with few edges for $n$-vertex
graphs having a tree-decomposition into bags of diameter at most
$\lambda$, i.e., the tree-length $\lambda$ graphs. For such graphs
they construct additive $2\lambda$-spanners with $O(\lambda n+n \log
n)$ edges, and additive $4\lambda$-spanners with $O(\lambda n)$
edges. Combining these results with the results of
\cite{APPROX2011}, we
obtain the following interesting fact (
in a sense,  turning a multiplicative stretch into an additive
surplus without much increase in the number of edges).

\begin{theorem} \label{th:easy} (combining \cite{DoDrGaYa} and \cite{APPROX2011}) If a graph $G$ admits a (multiplicative) tree
$t$-spanner then it has an additive $2t$-spanner with $O(t n+n \log
n)$ edges  and an additive $4t$-spanner with $O(tn)$ edges, both
constructible in polynomial time.
\end{theorem}

This fact rises few intriguing questions.
Does a polynomial time algorithm exist that, given an $n$-vertex
graph $G$
admitting a (multiplicative) tree $t$-spanner, constructs an
additive $O(t)$-spanner of $G$ with $O(n)$ or $O(n\log n)$ edges
(where the number of edges in the spanner is independent of $t$)? Is
a result similar to one presented by Elkin and Peleg in \cite{ElPe}
possible? Namely, does a polynomial time algorithm exist that, given
an $n$-vertex graph $G$ admitting a (multiplicative) tree
$t$-spanner, constructs an additive $(t-1)$-spanner\footnote{Recall
that any additive $(t-1)$-spanner is a multiplicative $t$-spanner.}
of $G$ with $O(n\log n)$
edges? 
If we allow to use more trees (like in collective tree spanners),
does a polynomial time algorithm exist that, given an $n$-vertex
graph $G$ admitting a (multiplicative) tree $t$-spanner, constructs
a system of $\tilde{O}(1)$ collective additive tree
$\tilde{O}(t)$-spanners of $G$ (where $\tilde{O}$ is similar to
Big-$O$ notation up to a poly-logarithmic factor)?
%
%
Note that an interesting question whether a multiplicative tree
spanner can be turned into an additive tree spanner with a slight
increase in the stretch is (negatively) settled already in
\cite{EmekP04}: if there exist some $\delta=o(n)$ and $\epsilon>0$
and a polynomial time algorithm that for any graph admitting a tree
$t$-spanner constructs a tree $((6/5-\epsilon)t+\delta)$-spanner,
then P=NP.

We give some partial answers to these questions in Section
\ref{sec:k=1}.
We investigate there a more general  question whether a graph with
bounded tree-breadth admits a small system of collective additive
tree spanners. We show that any $n$-vertex graph $G$
has a system of at most $\log_2 n$ collective additive tree
$(2\rho\log_2 n)$-spanners, where $\rho\leq \tb(G)$. This settles
also an open question from \cite{DoDrGaYa} whether a graph with
tree-length $\lambda$ admits a small system of collective additive
tree $\tilde{O}(\lambda)$-spanners.

As a consequence, we obtain that there is a polynomial time
algorithm that, given an $n$-vertex graph $G$ admitting a
(multiplicative) tree $t$-spanner, constructs:
\begin{enumerate}\vspace*{-2mm}
   \item[-] a system of at most $\log_2 n$ collective additive tree $O(t\log n)$-spanners of
   $G$ (compare with \cite{APPROX2011,EmekP04} where a multiplicative tree $O(t\log n)$-spanner
   was constructed for $G$ in polynomial time; thus, we ``have turned" a multiplicative tree $O(t\log n)$-spanner
   into at most $\log_2 n$ collective additive tree $O(t\log n)$-spanners);
  \item[-] an additive $O(t\log n)$-spanner of $G$ with at most $n\log_2 n$
  edges (compare with Theorem \ref{th:easy}).
\end{enumerate}\vspace*{-2mm}

In Section \ref{sec:any-k} we generalize the method of Section
\ref{sec:k=1}. We define a new notion which combines both the
tree-width and the tree-breadth of a graph.

 The \emph{$k$-breadth} of a tree-decomposition $T(G)=\left(\{X_i| i\in I\}, T=(I,F)\right)$
of a graph $G$ is the minimum integer $r$ such that for each bag
$X_i,i\in I$, there is a set of at most $k$ vertices $C_i=\{v^i_j|
v^i_j \in V(G), j=1,\ldots,k\}$ such that for each $u\in X_i$, we
have $d_G(u,C_i)\leq r$ (i.e., each bag $X_i$ can be covered with at
most $k$ disks of $G$ of radius at most $r$ each; $X_i \subseteq
D_r(v^i_1,G)\cup \ldots \cup D_r(v^i_k,G)$).
%
The \emph{$k$-tree-breadth } of a graph $G$, denoted by $\tb_k(G)$,
is the minimum of the $k$-breadth, over all tree-decompositions of
$G$. We say that a family of graphs $\mathcal{G}$ is {\em of bounded
$k$-tree-breadth,} if there is a constant $c$ such that for each
graph $G$ from $\mathcal{G}$, $\tb_k(G)\leq c$.
Clearly, for every graph $G$, $\tb(G)=\tb_1(G)$, and $\tw(G)\leq
k-1$ if and only if $\tb_k(G)=0$. Thus, the notions of the
tree-width and the tree-breadth are particular cases of the
$k$-tree-breadth.

In Section \ref{sec:any-k}, we show that any $n$-vertex graph $G$
with $\tb_k(G)\leq \rho$ has a system of at most $k(1+ \log_2 n)$
collective additive tree $(2\rho(1+\log_2 n))$-spanners. In Section
\ref{sec:conseq}, we extend a result from  \cite{APPROX2011} and
show that if a graph $G$ admits a (multiplicative) $t$-spanner $H$
with $\tw(H)=k-1$ then its $k$-tree-breadth is at most
$\lceil{t/2}\rceil$. As a consequence, we obtain that, for every
fixed $k$, there is a polynomial time algorithm that, given an
$n$-vertex graph $G$ admitting a (multiplicative) $t$-spanner with
tree-width at most $k-1$, constructs:
\begin{enumerate}\vspace*{-2mm}
   \item[-] a system of at most $k(1+ \log_2 n)$ collective additive tree $O(t\log n)$-spanners of
   $G$;
  \item[-] an additive $O(t\log n)$-spanner of $G$ with at most $O(k n\log n)$
  edges.
\end{enumerate}\vspace*{-2mm}

 We conclude the paper with few open questions.

%
\section{Preliminaries}\label{section:definition}
All graphs occurring in this paper are connected, finite,
unweighted, undirected, loopless and without multiple edges. We call
$G=(V,E)$ an {\em $n$-vertex $m$-edge graph} if $|V|=n$ and $|E|=m$.
A {\em clique} is a set of pairwise adjacent vertices of $G$. By
$G[S]$ we denote a subgraph of $G$ induced by vertices of
$S\subseteq V$. Let also $G\setminus S$ be the graph $G[V\setminus
S]$ (which is not necessarily connected). A set $S\subseteq V$ is
called a {\em separator} of $G$ if the graph $G[V\setminus S]$ has
more than one connected component, and $S$ is called a {\em balanced
separator} of $G$ if each connected component of  $G[V\setminus S]$
has at most $|V|/2$ vertices. A set $C\subseteq V$ is called a {\em
balanced clique-separator} of $G$ if $C$ is both a clique and a
balanced separator of $G$. For a vertex $v$ of $G$, the sets
$N_G(v)=\{w\in V| vw\in E\}$ and $N_G[v]=N_G(v)\cup\{v\}$ are called
the \emph{open neighborhood} and the \emph{closed neighborhood} of
$v$, respectively.


In a graph $G$ the {\em length} of a path from a vertex $v$ to a
vertex $u$ is the number of edges in the path. The {\em distance}
$d_G(u,v)$ between vertices $u$ and  $v$ is the length of a shortest
path connecting $u$ and $v$ in $G$. The {\em diameter} in $G$ of a
set $S\subseteq V$ is $\max_{x,y\in S}d_G(x,y)$ and its {\em radius}
in $G$ is $\min_{x\in V}\max_{y\in S}$ $d_G(x,y)$ (in some papers
they are called the {\em weak diameter} and the {\em weak radius} to
indicate that the distances are measured in $G$ not in $G[S]$). The
\emph{disk} of $G$ of radius $k$ centered at vertex $v$ is the set
of all vertices at distance at most $k$ to $v$: $D_k(v,G)=\{w\in V|
d_G(v, w)\leq k\}.$ A disk $D_k(v,G)$ is called a {\em balanced
disk-separator} of $G$ if the set $D_k(v,G)$ is a balanced separator
of $G$.

It is easy to see that the $t$-spanners can equivalently be defined
as follows.
\begin{proposition}\label{prop:zero}
Let $G$ be a connected graph and $t$ be a positive number. A
spanning subgraph $H$ of $G$ is a $t$-spanner of $G$ if and only if
for every edge $xy$ of $G$, $d_H(x,y)\leq t$ holds.
\end{proposition}
This proposition implies that the  stretch of a spanning subgraph of
a graph $G$ is always obtained on a pair of vertices that form an
edge in $G$. Consequently, throughout this paper, $t$ can be
considered as an integer which is greater than 1 (the case $t=1$ is
trivial since $H$ must be $G$ itself).

It is also known that every additive $r$-spanner of $G$ is a
(multiplicative) $(r+1)$-spanner of $G$.
\begin{proposition}\label{prop:zero+}
Every additive $r$-spanner of $G$ is a (multiplicative)
$(r+1)$-spanner of $G$. The converse is generally not true.
\end{proposition}

\section{Collective Additive Tree Spanners and the Tree-Breadth of a Graph} \label{sec:k=1}
In this section, we show that every $n$-vertex graph $G$
has a system of at most $\log_2 n$ collective additive tree
$(2\rho\log_2 n)$-spanners, where $\rho\leq \tb(G)$. We also discuss
consequences of this result. Our method is a generalization of
techniques used in \cite{CollTrSp} and \cite{APPROX2011}. We will
assume that $n\geq 4$ since any connected graph with at most 3
vertices has an additive tree 1-spanner.


Note that we do not assume here that a tree-decomposition $T(G)$ of
breadth $\rho$ is given for $G$ as part of the input. Our method
does not need to know $T(G)$, our algorithm works directly on $G$.
For a given graph $G$ and an integer $\rho$, even checking whether
$G$ has a tree-decomposition of breadth $\rho$ could be a hard
problem. For example, while graphs with tree-length 1 (as they are
exactly the chordal graphs) can be recognized in linear time, the
problem of determining whether a given graph has tree-length at most
$\lambda$ is NP-complete for every fixed $\lambda >1$ (see
\cite{Daniel10}).

We will need the following  results proven in \cite{APPROX2011}.

\begin{lemma}[\cite{APPROX2011}] \label{lm:gch-sep}
Every graph $G$ 
has a balanced disk-separator $D_r(v,G)$ centered at some vertex
$v$, where $r\leq \tb(G)$.
\end{lemma}


\begin{lemma}[\cite{APPROX2011}] \label{prop:gch-sep-time}
For an arbitrary graph $G$ with $n$ vertices and $m$ edges a
balanced disk-separator $D_r(v,G)$   with minimum $r$ can be found
in $O(n m)$ time.
\end{lemma}


\subsection{Hierarchical decomposition of a  graph with bounded
tree-breadth}\label{ssec:dec-tree} In this subsection, following
\cite{APPROX2011}, we show how to decompose a graph with bounded
tree-breadth and build a hierarchical decomposition tree for it.
This hierarchical decomposition tree is used later for construction
of collective additive tree spanners for such a graph.

Let $G=(V,E)$ be an arbitrary connected $n$-vertex $m$-edge graph
with a disk-separator $D_r(v,G)$. Also, let $G_1,\dots,G_q$ be the
connected components of $G[V\setminus D_{r}(v,G)]$. Denote by
$S_i:=\{x\in V(G_i)| d_G(x,D_{r}(v,G))=1\}$ the neighborhood of
$D_{r}(v,G)$ with respect to $G_i$. Let also $G_i^+$ be the graph
obtained from component $G_i$ by adding a vertex $c_i$ ({\em
representative} of $D_{r}(v,G)$) and making it adjacent to all
vertices of $S_i$, i.e., for a vertex $x\in V(G_i)$, $c_ix\in
E(G_i^+)$ if and only if there is a vertex $x_D\in D_{r}(v,G)$ with
$xx_D\in E(G)$. See Fig. \ref{fig:ch-dec} for an illustration. In
what follows, we will call vertex $c_i$ a {\em meta vertex
representing disk $D_r(v,G)$ in graph $G_i^+$}.  Given a graph $G$
and its disk-separator $D_r(v,G)$, the graphs $G_1^+,\dots,G_q^+$
can be constructed in total time $O(m)$. Furthermore, the total
number of edges in the graphs $G_1^+,\dots,G_q^+$ does not exceed
the number of edges in $G$, and the total number of vertices
(including $q$ meta vertices) in those graphs does not exceed the
number of vertices in $G[V\setminus D_{r}(v,G)]$ plus $q$.

\begin{figure}[htb]
\begin{center} \vspace*{-5mm}
 \begin{minipage}[b]{15cm}
    \begin{center} \hspace*{10mm}
\includegraphics[height=5.cm]{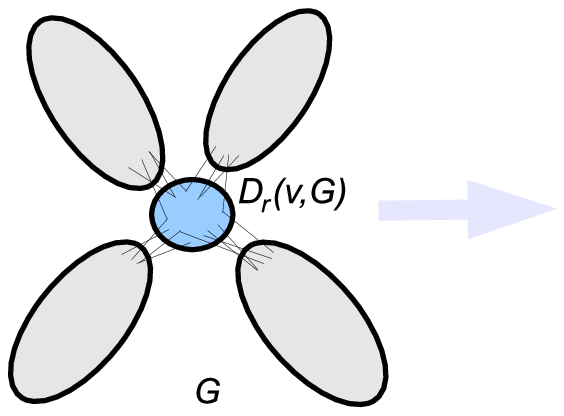}\hspace*{-25mm}
\includegraphics[height=5.cm]{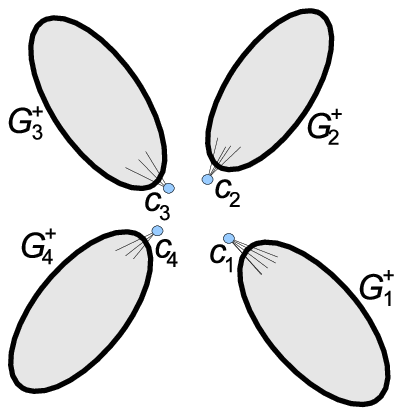}
    \end{center} \vspace*{-18mm}
    \caption{\label{fig:ch-dec} A graph $G$ with a disk-separator $D_r(v,G)$ and the corresponding graphs $G^+_1,\dots, G^+_4$ obtained from $G$.
    $c_1,\dots,c_4$ are meta vertices representing the disk $D_r(v,G)$ in the corresponding graphs. } %
 \end{minipage}
\end{center}
\vspace*{-6mm}
\end{figure}

Denote by $G_{/e}$ the graph obtained from $G$ by contracting its
edge $e$. Recall that edge $e$ contraction is an operation which
removes $e$ from $G$ while simultaneously merging together the two
vertices $e$ previously connected.  If a contraction  results in
multiple edges, we delete duplicates of an edge to stay within the
class of simple graphs. The operation may be performed on a set of
edges by contracting each edge (in any order). The following lemma
guarantees that the tree-breadths of the graphs $G^+_i$,
$i=1,\dots,q$,  are no larger than the tree-breadth of $G$.

\begin{lemma}[\cite{APPROX2011}]\label{lem:G+hereditary}
For any graph $G$ and its edge $e$, $\tb(G)\leq \rho$ implies
$\tb(G_{/e})\leq \rho$. Consequently, for any graph $G$ with
$\tb(G)\leq \rho$,   $\tb(G_i^+)\leq \rho$ holds for each
$i=1,\dots,q$.
\end{lemma}

Clearly, one can get $G_i^+$ from $G$ by repeatedly contracting (in
any order) edges of $G$ that are not incident to vertices of $G_i$.
In other words, $G_i^+$ is a minor of $G$. Recall that a graph $G'$
is a {\em minor} of $G$ if $G'$ can be obtained from $G$ by
contracting some edges, deleting some edges, and deleting some
isolated vertices. The order in which a sequence of such
contractions and deletions is performed on $G$ does not affect the
resulting graph $G'$.

Let $G=(V,E)$ be a connected  $n$-vertex, $m$-edge graph and assume
that $\tb(G)\leq\rho$. Lemma \ref{lm:gch-sep} and Lemma
\ref{prop:gch-sep-time} guarantee that $G$ has a balanced
disk-separator $D_{r}(v,G)$ with $r\leq \rho$, which can be found in
$O(n m)$ time by an algorithm that works directly on graph $G$ and
does not require construction of a tree-decomposition of $G$ of
breadth $\leq \rho$.  Using these and Lemma \ref{lem:G+hereditary},
we can build a (rooted) {\em hierarchical tree} $\cH(G)$ for $G$ as
follows. If $G$ is a connected graph with at most 5 vertices,
then $\mathcal{H}(G)$ is one node tree with root node 
$(V(G), G)$. 
Otherwise, find a balanced disk-separator $D_{r}(v,G)$ in $G$ with
minimum $r$ (see Lemma ~\ref{prop:gch-sep-time}) and construct the
corresponding graphs $G_1^+, G_2^+,\dots,G_q^+$. For each graph
$G_i^+$ ($i=1,\dots,q$) (by Lemma \ref{lem:G+hereditary},
$\tb(G_i^+)\leq \rho$), construct a hierarchical tree
$\mathcal{H}(G_i^+)$ recursively and build $\mathcal{H}(G)$ by
taking the pair $(D_{r}(v,G),G)$ to be the root and connecting the
root of each tree $\mathcal{H}(G_i^+)$ as a child of
$(D_{r}(v,G),G)$.

The depth of this tree $\cH(G)$ is the smallest integer $k$ such
that $$\frac{n}{2^k}+\frac{1}{2^{k-1}}+\dots +\frac{1}{2}+1\leq 5,$$
that is, the depth 
is at most $\log_2 n -1$.
\commentout{
{\bf ====== Calculations  =====}

Since $$\frac{n}{2^{k-1}}+\frac{1}{2^{k-2}}+\dots +\frac{1}{2}+1\geq
6$$ and $$\frac{1}{2^{k-2}}+\dots +\frac{1}{2}+1<2,$$ we get
$$\frac{n}{2^{k-1}}>6-2=4,$$ $$n> 2^{k+1},$$   $$k+1< \log_2 n,$$
$$k< \log_2 n -1.$$

{\bf ========================} }
%
%

It is also easy to see that, given a graph $G$ with $n$ vertices and
$m$ edges, a hierarchical tree $\cH(G)$
can be constructed in $O(nm\log^2 n)$ total time. There are  at most
$O(\log n)$ levels in  $\cH(G)$, and one needs to do at most
$O(nm\log n)$ operations per level since the total number of edges
in the graphs of each level is at most $m$ and the total number of
vertices in those graphs can not exceed $O(n\log n)$.

For an internal (i.e., non-leaf) node $Y$ of $\mathcal{H}(G)$, since
it is associated with a pair $(D_{r'}(v',G'),G')$, where $r'\leq
\rho$, $G'$ is a minor of $G$ and $v'$ is the center of disk
$D_{r'}(v',G')$ of $G'$, it will be convenient, in what follows, to
denote $G'$ by $G(\downarrow Y)$, $v'$ by $c(Y)$, $r'$ by $r(Y)$,
and $D_{r'}(v',G')$ by $Y$ itself. Thus,
$(D_{r'}(v',G'),G')=(D_{r(Y)}(c(Y),G(\downarrow Y)),G(\downarrow
Y))=(Y, G(\downarrow Y))$ in these notations, and we identify node
$Y$ of $\mathcal{H}(G)$ with the set $Y=D_{r(Y)}(c(Y),G(\downarrow
Y))$ and associate with this node also the graph $G(\downarrow Y)$.
See Fig. \ref{fig:dec-tree} for an illustration.
Each leaf $Y$ of $\mathcal{H}(G)$, since it 
corresponds to a pair $(V(G'),G')$, we identify with the set
$Y=V(G')$ and use, for a convenience, the notation $G(\downarrow Y)$
for $G'$.

\begin{figure}[htb]
\begin{center} \vspace*{-20mm}
 \begin{minipage}[b]{16cm}
    \begin{center} \hspace*{5mm}
\includegraphics[height=8.cm]{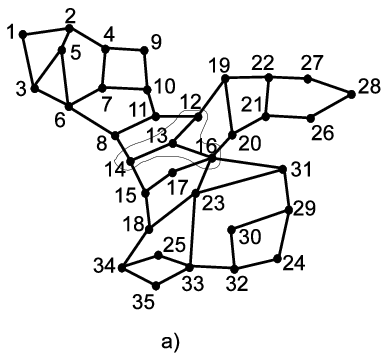}\hspace*{-60mm}
\includegraphics[height=10.cm]{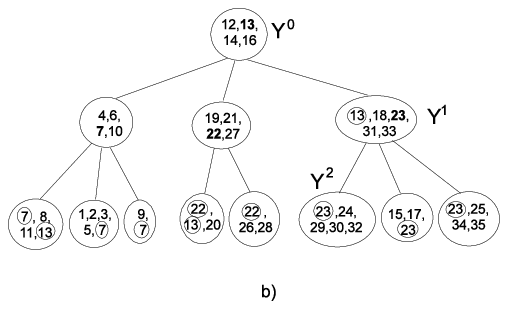}
    \end{center} \vspace*{-52mm}
    \begin{center} \hspace*{30mm}
\includegraphics[height=8.cm]{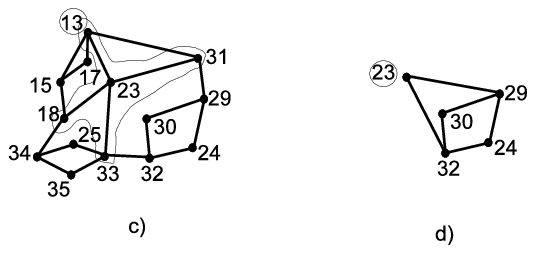}
    \end{center} \vspace*{-30mm}
    \caption{\label{fig:dec-tree} a) A graph $G$ and its balanced disk-separator $D_1(13,G)$. b) A hierarchical tree $\cH(G)$ of $G$. We have $G=G(\downarrow Y^0)$,
    $Y^0=D_1(13,G)$. Meta vertices are shown circled, disk centers are shown in bold.
    c) The graph $G(\downarrow Y^1)$ with its balanced disk-separator $D_1(23,G(\downarrow Y^1))=Y^1$. $G(\downarrow Y^1)$ is a minor of $G(\downarrow Y^0)$.
    d) The graph $G(\downarrow Y^2)$, a minor of $G(\downarrow Y^1)$ and of $G(\downarrow Y^0)$. $Y^2=V(G(\downarrow Y^2))$ is a leaf of $\cH(G)$.} %
 \end{minipage}
\end{center}
\vspace*{-6mm}
\end{figure}

If now $(Y^0,Y^1,\dots, Y^h)$ is the path of $\cH(G)$ connecting the
root $Y^0$ of $\cH(G)$ with a node $Y^h$, then the vertex set of the
graph $G(\downarrow Y^h)$ consists of some (original) vertices of
$G$ plus at most $h$ meta vertices representing the disks
$D_{r(Y)}(c(Y^i),G(\downarrow Y^i))=Y^i$, $i=0,1,\dots,h-1$. Note
also that each (original) vertex of $G$ belongs to exactly one node
of $\mathcal{H}(G)$.


\subsection{Construction of collective additive tree spanners}
Unfortunately, the class of graphs of bounded tree-breadth is not
hereditary, i.e., induced subgraphs of a graph with tree-breath
$\rho$ are not necessarily of tree-breadth at most $\rho$ (for
example, a cycle of length $\ell$ with one extra vertex adjacent to
each vertex of the cycle has tree-breadth 1, but the cycle itself
has tree-breadth $\ell/3$). Thus, the method presented in
\cite{CollTrSp}, for constructing collective additive tree spanners
for hereditary classes of graphs admitting balanced disk-separators,
cannot be applied directly to the graphs of bounded tree-breadth.
Nevertheless, we will show that, with the help of Lemma
\ref{lem:G+hereditary}, the notion of hierarchical tree from
previous subsection and a careful analysis of distance changes (see
Lemma \ref{lem:closest-center}), it is possible to generalize the
method of \cite{CollTrSp} and construct in polynomial time for every
$n$-vertex graph $G$ a system of at most $\log_2 n$ collective
additive tree $(2\rho\log_2 n)$-spanners, where $\rho\leq \tb(G)$.
Unavoidable presence of meta vertices in the graphs resulting from a
hierarchical decomposition of the original graph $G$ complicates the
construction and the analysis. Recall that, in \cite{CollTrSp}, it
was shown that if every induced subgraph of a graph $G$ enjoys a
balanced disk-separator with radius at most $r$, then $G$ admits a
system of at most $\log_2 n$ collective additive tree $2r$-spanners.



Let $G=(V,E)$ be a connected  $n$-vertex, $m$-edge graph and assume
that $\tb(G)\leq\rho$. Let $\cH(G)$ be a hierarchical tree of $G$.
Consider an arbitrary internal node  $Y^h$ of  $\cH(G)$, and let
$(Y^0,Y^1,\dots, Y^h)$ be the path of $\cH(G)$ connecting
the root $Y^0$ of $\cH(G)$ with $Y^h$. 
Let $\widehat{G}(\dY^j)$ be the graph obtained from $G(\dY^j)$ by
removing all its meta vertices (note that $\widehat{G}(\dY^j)$ may
be disconnected).

\begin{lemma}\label{lem:closest-center}
For any vertex $z$ from $Y^h\cap V(G)$  there exists an index $i\in
\{0,1,\dots,h\}$ such that the vertices $z$ and $c(Y^i)$ can be
connected in the graph $\widehat{G}(\downarrow Y^i)$ by a path of
length at most $\rho(h+1)$. In particular, $d_G(z,c(Y^i))\leq
\rho(h+1)$ holds.
%
%
\end{lemma}

\begin{proof}
Set $G_h:=G(\downarrow Y^h)$, $c:=c(Y^h)$, and let $SP^{G_h}_{c,z}$
be a shortest path of $G_h$ connecting vertices $c$ and $z$. We know
that this path has at most $r(Y^h)\leq \rho$ edges. If
$SP^{G_h}_{c,z}$ does not contain any meta vertices, then this path
is a path of $\widehat{G}(\downarrow Y^h)$ and of $G$ and therefore
$d_G(c,z)\leq \rho$ holds.

Assume now that $SP^{G_h}_{c,z}$ does contain meta vertices and let
$\mu'$ be the closest to $z$ meta vertex in $SP_{c,z}^{G_h}$. See
Fig.~\ref{fig:cc} for an illustration. Let
$SP^{G_h}_{c,z}=(c,\dots,a',\mu',b',\dots,z)$. By construction of
$\cH(G)$, meta vertex $\mu'$ was created at some earlier recursive
step to represent disk $Y^{i'}$ of graph $G_{i'}:=G(\downarrow
Y^{i'})$ for some 
$i'\in \{0,\dots,h-1\}$. Hence, there is a path
$P^{G_{i'}}_{c',z}=(c',\dots,b',\dots,z)$ of length at most $2\rho$
in $G_{i'}$ with $c':=c(Y^{i'})$. Again, if $P^{G_{i'}}_{c',z}$ does
not contain any meta vertices, then this path is a path of
$\widehat{G}(\downarrow Y^{i'})$ and of $G$ and therefore
$d_G(c',z)\leq 2\rho$ holds. If $P^{G_{i'}}_{c',z}$ does contain
meta vertices then again, ``unfolding" a meta vertex $\mu''$ of
$P^{G_{i'}}_{c',z}$ closest to $z$, we obtain a path
$P^{G_{i''}}_{c'',z}$ of length at most $3\rho$ in
$G_{i''}:=G(\downarrow Y^{i''})$ with $c'':=c(Y^{i''})$ for some
$i''\in \{0,\dots,i'-1\}$.

By continuing ``unfolding" this way meta vertices closest to $z$,
after at most $h$ steps, we will arrive at the situation when,  for
some index $i^*\in \{0,1,\dots,h\}$, a path of length at most
$\rho(h+1)$ will connect vertices $z$ and $c(Y^{i^*})$ in the graph
$\widehat{G}(\downarrow Y^{i^*})$. \qed

\begin{figure}[htb]
\begin{center} 
 \begin{minipage}[b]{14cm}
    \begin{center} \hspace*{-10mm}
\includegraphics[height=5.cm]{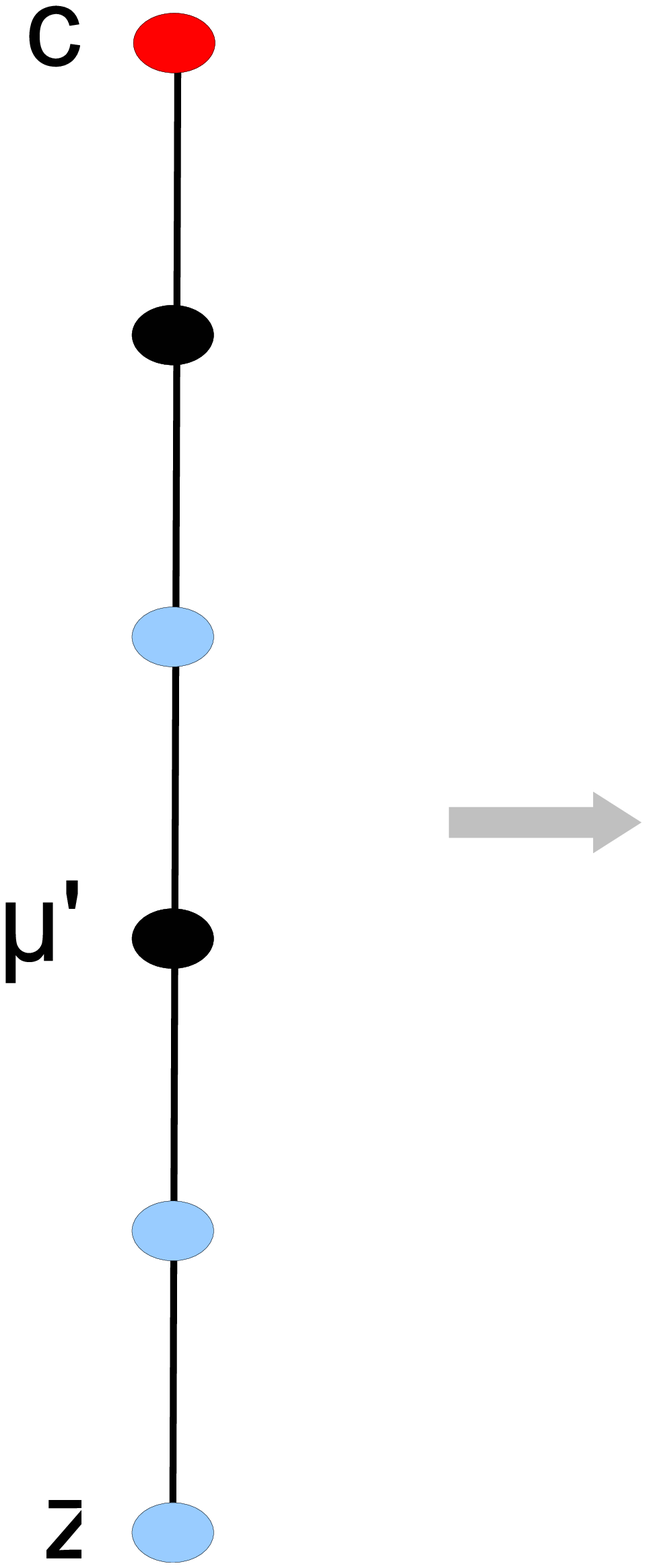}
\includegraphics[height=5.cm]{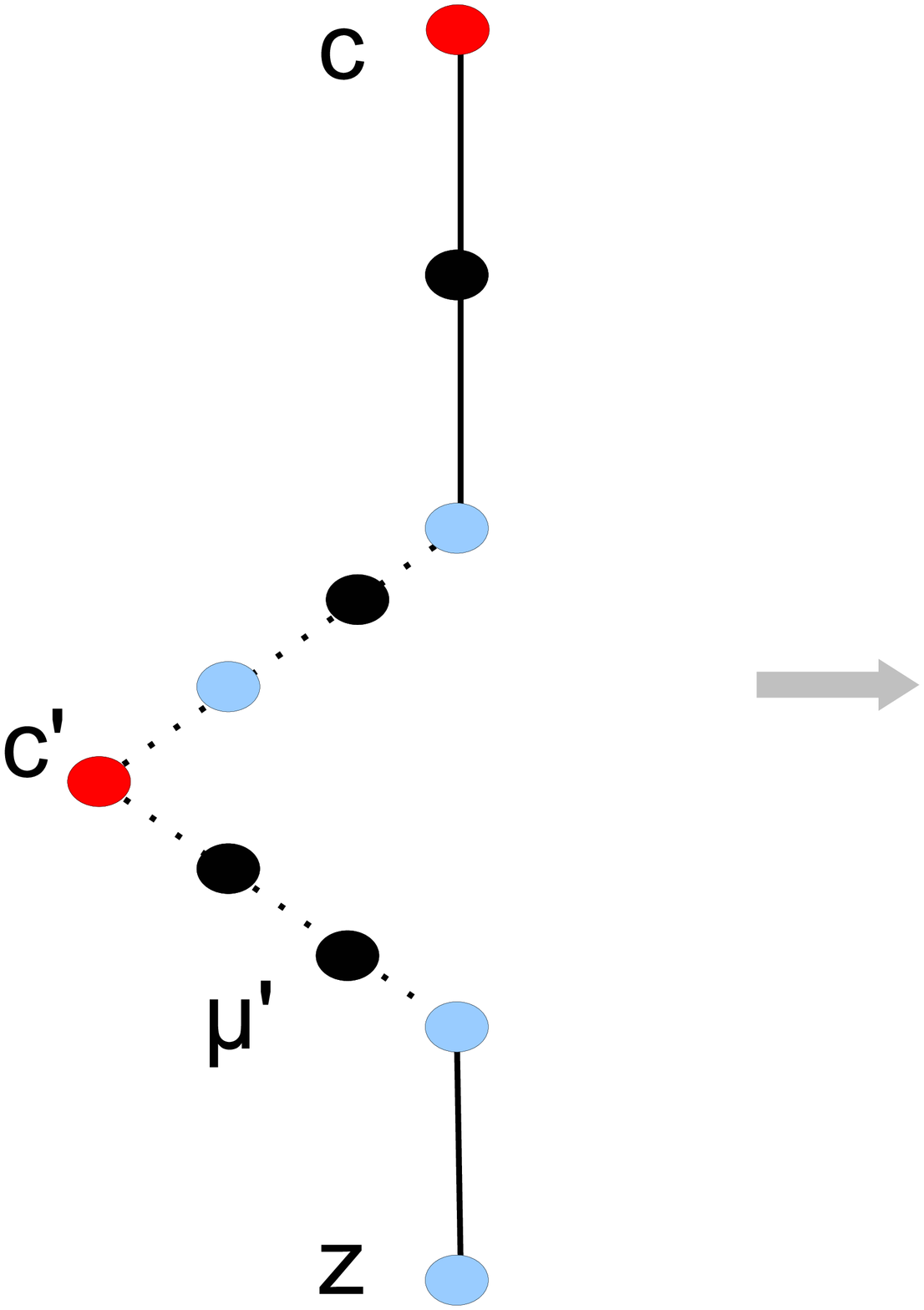}\hspace*{-5mm}
\includegraphics[height=5.cm]{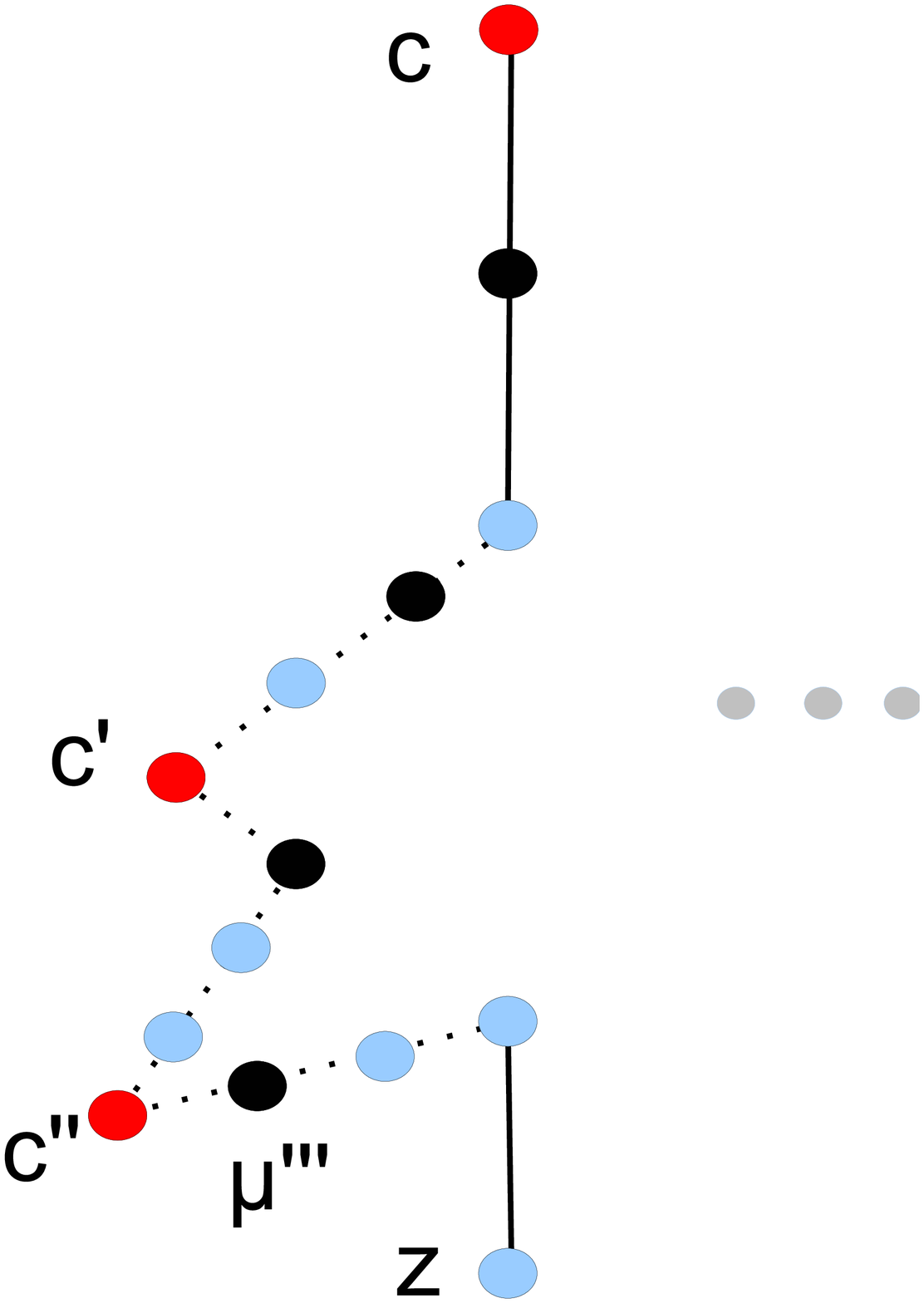}\hspace*{5mm}
\includegraphics[height=5.cm]{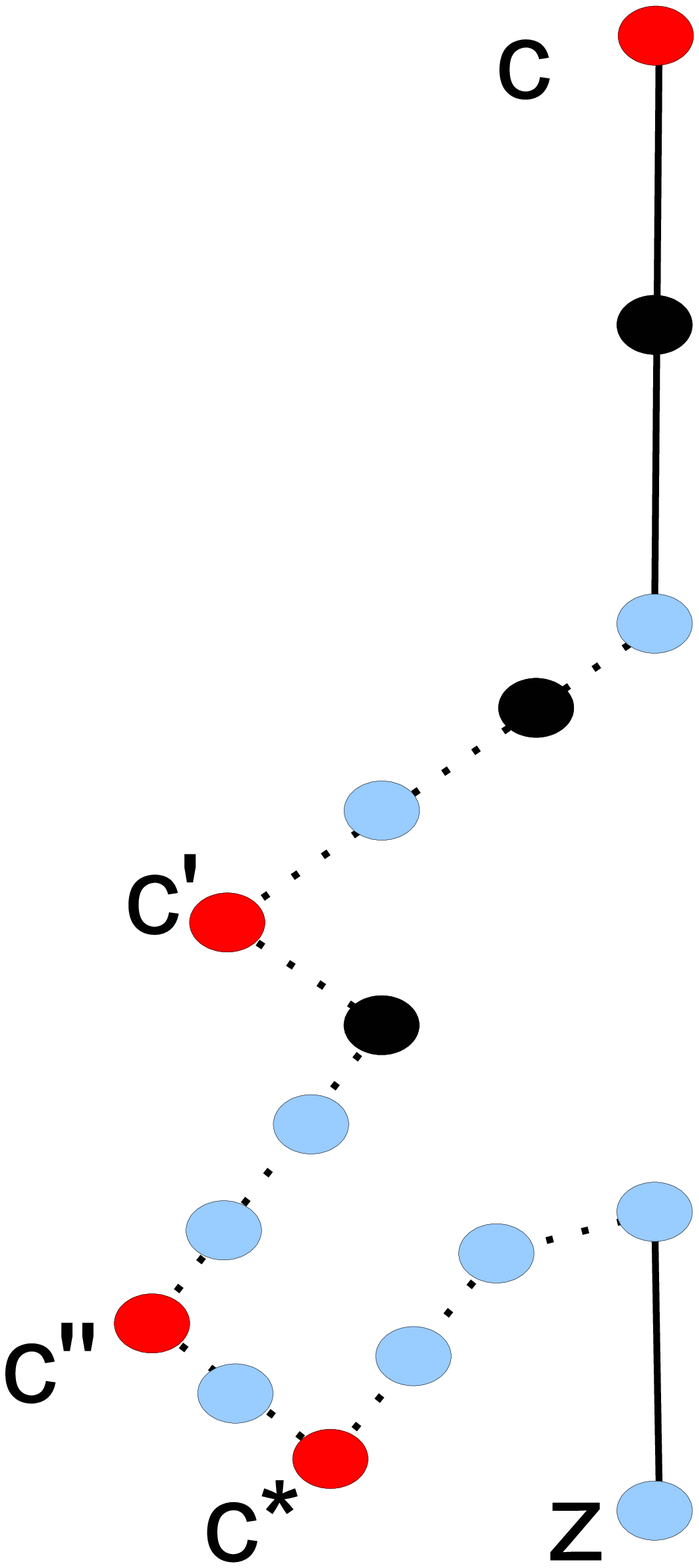}
    \end{center} \vspace*{-5mm}
    \caption{\label{fig:cc} Illustration to the proof of Lemma \ref{lem:closest-center}: ``unfolding" meta vertices. } %
 \end{minipage}
\end{center}
\vspace*{-6mm}
\end{figure}

\end{proof}

Consider two arbitrary vertices $x$ and $y$ of $G$, and let $S(x)$
and $S(y)$ be the nodes of $\cH(G)$ containing $x$ and $y$,
respectively. Let also $NCA_{\cH(G)}(S(x),S(y))$ be the nearest
common ancestor of nodes $S(x)$ and $S(y)$ in $\cH(G)$ and
$(Y^0,Y^1,\dots, Y^h)$ be the path of $\cH(G)$ connecting the root
$Y^0$ of $\cH(G)$ with $NCA_{\cH(G)}(S(x),S(y))=Y^h$ (in other
words, $Y^0,Y^1,\dots,Y^h$ are the common ancestors of $S(x)$ and
$S(y)$).

\begin{lemma} \label{lm:decomp}
Any path $P^G_{x,y}$ connecting vertices $x$ and $y$ in $G$ contains
a vertex from  $Y^0\cup Y^1\cup \dots \cup Y^h$.
\end{lemma}

Let $SP^G_{x,y}$ be a shortest path of $G$ connecting vertices $x$
and $y$, and  let $Y^i$ be the node of the path
$(Y^0,Y^1,\dots,Y^h)$ with the smallest index such that
$SP^G_{x,y}\bigcap Y^i\neq \emptyset$ in $G$. The following lemma
holds.

\begin{lemma} \label{lm:isom}
For each $j=0,\dots,i$, we have $d_G(x,y)=d_{G'}(x,y)$, where
$G':=\widehat{G}(\dY^j)$.
\end{lemma}


Let now $B^i_1,\dots, B^i_{p_i}$ be the nodes at depth $i$ of the
tree $\cH(G)$. For each node $B^i_j$ that is not a leaf of $\cH(G)$,
consider its (central) vertex $c_j^i:=c(B_j^i)$. If $c_j^i$ is an
original vertex of $G$ (not a meta vertex created during the
construction of $\cH(G)$), then define a  connected graph $G_j^i$
obtained from $G(\downarrow B_j^i)$ by removing all its meta
vertices. If removal of those meta vertices produced few connected
components, choose as $G_j^i$ that component which contains the
vertex $c_j^i$. Denote by $T_j^i$ a BFS--tree of graph $G_j^i$
rooted at vertex $c_j^i$ of
$B_j^i$. 
If $B^i_j$ is a leaf of $\cH(G)$, then $B^i_j$ has at most 5
vertices. In this case, remove all meta vertices from $G(\downarrow
B_j^i)$ and for each connected component of the resulting graph
construct an additive tree spanner with optimal surplus $\leq 3$.
Denote the resulting subtree (forest) by $T_j^i$.

The trees $T_j^i$ ($i=0,1,\dots,depth(\cH(G))$, $j=1,2,\dots,p_i$),
obtained this way, are called {\em local subtrees} of $G$. Clearly,
the construction of these local subtrees can be incorporated into
the procedure of constructing hierarchical tree $\cH(G)$ of $G$ and
will not increase the overall $O(nm\log^2 n)$ run-time (see
Subsection \ref{ssec:dec-tree}).

\begin{lemma} \label{lm:distT}
For any two vertices $x,y \in V(G)$, there exists a local subtree
$T$ such that $d_{T}(x,y)\leq d_{G}(x,y)+2\rho \log_2 n - 1$.
\end{lemma}

\begin{proof}
We know, by Lemma \ref{lm:isom}, that a shortest path
$SP_{x,y}^{G}$, intersecting $Y^i$ and not intersecting any $Y^l$
($l<i$), lies entirely in $G':=\widehat{G}(\downarrow Y^i)$. Thus,
$d_G(x,y)=d_{G'}(x,y)$. If $Y^i$ is a leaf of $\cH(G)$ then for a
local subtree $T'$ (it could be a forest) of $G$ constructed for
$G'$ the following holds: $d_{T'}(x,y)\leq
d_{G'}(x,y)+3=d_{G}(x,y)+3\leq d_{G}(x,y)+2\rho\log_2 n-1$ (since
$n\geq 4$ and $\rho\geq 1$).

Assume now that $Y^i$ is an internal node of $\cH(G)$. We have
$i\leq \log_2 n -2$, since the depth of $\cH(G)$ is at most $\log_2
n -1$. Let $z \in Y^i$ be a vertex on the shortest path
$SP_{x,y}^{G}$. By Lemma \ref{lem:closest-center}, there exists an
index $j\in \{0,1,\dots,i\}$ such that the vertices $z$ and $c(Y^j)$
can be connected in the graph $\widehat{G}(\downarrow Y^j)$ by a
path of length at most $\rho(i+1)$.
Set
$G'':=\widehat{G}(\downarrow Y^j)$ and $c:=c(Y^j)$. By Lemma
\ref{lm:isom}, $d_G(x,y)=d_{G'}(x,y)=d_{G''}(x,y)$. Let $T''$ be the
local tree constructed for graph $G''=\widehat{G}(\downarrow Y^j)$,
i.e., a BFS--tree of a connected component of the graph
$G''=\widehat{G}(\downarrow Y^j)$ and rooted at vertex $c=c(Y^j)$.

We have $d_{T''}(x,c)=d_{G''}(x,c)$ and $d_{T''}(y,c)=d_{G''}(y,c)$.
By the triangle inequality, $d_{T''}(x,c)=d_{G''}(x,c) \leq
d_{G''}(x,z)+d_{G''}(z,c)$ and $d_{T''}(y,c)=d_{G''}(y,c) \leq
d_{G''}(y,z)+d_{G''}(z,c).$ That is, $d_{T''}(x,y)\leq
d_{T''}(x,c)+d_{T''}(y,c) \leq
                   d_{G''}(x,z)+d_{G''}(y,z)+2d_{G''}(z,c) =d_{G''}(x,y)+2d_{G''}(z,c).$
Now, using Lemma ~\ref{lm:isom} and inequality $d_{G''}(z,c)\leq
\rho(i+1)\leq \rho(\log_2 n -1)$,  we get $d_{T''}(x,y)\leq
d_{G''}(x,y)+2d_{G''}(z,c)\leq d_{G}(x,y)+2\rho(\log_2 n -1).$ \qed
\end{proof}

This lemma implies two important results. Let $G$ be a graph with
$n$ vertices and $m$ edges having $\tb(G)\leq \rho$. Also, let
$\cH(G)$ be its hierarchical tree and $\mathcal{LT}(G)$ be the
family of all its local subtrees (defined above). Consider a graph
$H$ obtained by taking the union of all local subtrees of $G$ (by
putting all of them together), i.e.,
\[H:=\bigcup\{T_j^i| T_j^i\in \mathcal{LT}(G)\}=(V,\cup\{E(T_j^i)| T_j^i\in \mathcal{LT}(G)\}).\]

 Clearly, $H$ is a spanning subgraph of $G$, constructible in
$O(nm\log^2 n)$  total time, and, for any two vertices $x$ and $y$
of $G$, $d_H(x,y)\leq d_G(x,y)+2\rho \log_2 n - 1$ holds. Also,
since for every level $i$ ($i=0,1,\dots,depth(\cH(G))$) of
hierarchical tree  $\cH(G)$, the corresponding local subtrees
$T^i_1,\dots, T^i_{p_i}$ are pairwise vertex-disjoint, their union
has at most $n-1$ edges. Therefore, $H$ cannot have more than
$(n-1)\log_2 n$ edges in total. Thus, we have proven the following
result.

\begin{theorem} \label{tm:spanner}
Every graph $G$ with $n$ vertices and $\tb(G)\leq \rho$ admits an
additive $(2\rho\log_2 n)$--spanner with at most $n\log_2 n$ edges.
Furthermore, such a sparse additive spanner of $G$ can be
constructed in polynomial time.
\end{theorem}

Instead of taking the union of all local subtrees of $G$, one can
fix $i$ ($i\in\{0,1,\dots,depth(\cH(G))\}$) and consider separately
the union of only local subtrees $T^i_1,\dots, T^i_{p_i}$,
corresponding to the level $i$ of the hierarchical tree  $\cH(G)$,
and then extend in linear $O(m)$ time that forest to a spanning tree
$T^i$ of $G$ (using, for example, a variant of the Kruskal's
Spanning Tree algorithm for the unweighted graphs). We call this
tree $T^i$ the {\em spanning tree of $G$ corresponding to the level
$i$ of the hierarchical tree $\cH(G)$}. In this way we can obtain at
most $\log_2 n$ spanning trees for $G$, one for each level $i$ of
$\cH(G)$. Denote the collection of those spanning trees by $\cT(G)$.
Thus, we obtain the following theorem.

\begin{theorem} \label{tm:system}
Every graph $G$ with $n$ vertices and $\tb(G)\leq \rho$ admits a
system $\cT(G)$ of at most $\log_2 n$ collective additive tree
$(2\rho\log_2 n)$--spanners.  Furthermore, such a system of
collective additive tree spanners of $G$ can be constructed in
polynomial time.
\end{theorem}

\subsection{Additive spanners for graphs admitting (multiplicative) tree $t$--spanners}
Now we give two implications of the above results for the class of
tree $t$--spanner admissible graphs. In \cite{APPROX2011}, the
following important (``bridging") lemma was proven.

\begin{lemma}[\cite{APPROX2011}] \label{lm:tsp-tb} If a graph $G$ admits a tree $t$-spanner then its tree-breadth is at most $\lceil{t/2}\rceil$.
\end{lemma}

Note that the tree-breadth bounded by $\lceil{t/2}\rceil$ provides
only a necessary condition for a graph to have a multiplicative tree
$t$-spanner. There are (chordal) graphs which have tree-breadth 1
but any multiplicative tree $t$-spanner of them has $t=\Omega(\log
n)$ \cite{APPROX2011}. Furthermore, a cycle on $3n$ vertices has
tree-breadth $n$ but admits a system of 2 collective additive tree
$0$-spanners.

Combining Lemma ~\ref{lm:tsp-tb} with Theorem \ref{tm:spanner} and
Theorem ~\ref{tm:system}, we deduce the following results.

\begin{theorem}
Let $G$ be a graph with $n$ vertices and $m$ edges having a
(multiplicative) tree $t$--spanner. Then, $G$ admits an additive
$(2\lceil{t/2}\rceil\log_2 n)$--spanner with at most $n\log_2 n$
edges constructible in $O(nm\log^2n)$ time.
\end{theorem}

\begin{theorem}
Let $G$ be a graph with $n$ vertices and $m$ edges having a
(multiplicative) tree $t$--spanner. Then,  $G$ admits a system
$\cT(G)$ of at most $\log_2 n$ collective additive tree
$(2\lceil{t/2}\rceil\log_2 n)$--spanners constructible in
$O(nm\log^2 n)$ time.
\end{theorem}

\section{Collective Additive Tree Spanners of Graphs with Bounded $k$-Tree-Breadth, $k\geq 2$} \label{sec:any-k}
In this section, we extend the approach of Section~\ref{sec:k=1} and
show that any $n$-vertex graph $G$ with $\tb_k(G)\leq \rho$ has a
system of at most $k(1+ \log_2 n)$ collective additive tree
$(2\rho(1+\log_2 n))$-spanners constructible in polynomial time for
every fixed $k$. We will assume that $n>k$, since any graph with $n$
vertices has a system of $n-1$ collective additive tree 0-spanners
(consider $n-1$  BFS-trees rooted at different vertices).

\subsection{Balanced separators for graphs with bounded $k$-tree-breadth}

We will need the following balanced clique-separator result for
chordal graphs. Recall that a graph is chordal if every its induced
cycle has length three.


\begin{lemma}[\cite{GRE}]
\label{lm:ch-sep} Every chordal graph $G$ with $n$ vertices and $m$
edges contains a maximal clique $C$ such that if the vertices in $C$
are deleted from $G$, every connected component in the graph induced
by any remaining vertices is of size at most $n/2$. Such a balanced
clique-separator $C$ of a connected chordal graph can be found in
$O(m)$ time.
\end{lemma}

We say that a graph $G=(V,E)$ with $|V|\geq k$ has a \emph{balanced}
$\mathbf{D_r^k}$-\emph{separator} if there exists a collection of
$k$ disks $D_r(v_1,G), D_r(v_2,G),\dots, D_r(v_k,G)$ in $G$,
centered at (different) vertices $v_1,v_2,\dots,v_k$ 
and each of radius $r$, such that the union of those disks
$\mathbf{D_r^k}:=\bigcup_{i=1}^k D_r(v_i,G)$ forms a balanced
separator of $G$, i.e., each connected component of $G[V\setminus
\mathbf{D_r^k}]$ has at most $|V|/2$ vertices. The following result
generalizes Lemma \ref{lm:gch-sep}.

\begin{lemma}\label{disks_sep}
Every graph $G$ with at least $k$ vertices and $\tb_k(G) \leq \rho$
has a \emph{ balanced $\mathbf{D_\rho^k}$-separator}.
\end{lemma}

\begin{proof}
The proof of this lemma follows from {\em acyclic hypergraph}
theory. First we review some necessary definitions and an important
result characterizing acyclic hypergraphs. Recall that a {\em
hypergraph} $H$ is a pair $H = (V,\cal{E})$ where $V$ is a set of
vertices and $\cal{E}$ is a set of non-empty subsets of $V$ called
{\em hyperedges}.  For these and other hypergraph notions see
\cite{Be}.

Let $H=(V, \cal{E})$ be a {\em hypergraph} with the vertex set $V$
and the {\em hyperedge} set $\cal{E}$. 
For every vertex $v\in V$, let $\cal{E}$$(v)=\{e\in$ $\cal{E}$ $|
v\in e\}$.   The {\em 2--section graph} $2SEC(H)$ of a hypergraph
$H$ has $V$ as its vertex set and two distinct vertices are adjacent
in $2SEC(H)$ if and only if they are contained in a common hyperedge
of $H$. A hypergraph $H$
is called {\em conformal} if every 
clique 
of $2SEC(H)$ is contained in a hyperedge $e\in \cE$, and a
hypergraph $H$ is called {\em acyclic} if there is a tree $T$ with
node set $\cal E$ such that for all vertices $v \in V$,
$\cal{E}$$(v)$ induces a
subtree $T_v$ of $T$. %
It is a well-known fact (see, e.g., \cite{AAM,BFMY,Be}) that a
hypergraph $H$ is acyclic if and only if $H$ is conformal and
$2SEC(H)$ of $H$ is a chordal graph.

Let now $G=(V,E)$ be a graph with $\tb_k(G)=\rho$ and
${T}(G)=(\{X_i|i\in I\},{T}=(I,F))$ be its tree-decomposition of
$k$-breadth $\rho$. Evidently, the third condition of
tree-decompositions can be restated as follows: the hypergraph
$H=(V(G),\{X_i|i\in I\})$ is an acyclic hypergraph. Since each edge
of $G$ is contained in at least one bag of $T(G)$, the 2--section
graph $G^{*}:=2SEC(H)$ of $H$ is a chordal {\em supergraph} of the
graph $G$ (each edge of $G$ is an edge of $G^*$, but $G^*$ may have
some extra edges between non-adjacent vertices of $G$ contained in a
common bag of $T(G)$). By Lemma \ref{lm:ch-sep},  the chordal graph
$G^{*}$ contains a balanced clique-separator $C\subseteq V(G)$. By
conformality of $H$, $C$ must be contained in a bag of $T(G)$. From
the definition of $k$-breadth, there must exist $k$ vertices
$v_1,v_2,\dots,v_k$  such that $C \subseteq \mathbf{D_\rho^k}$,
where $\mathbf{D_\rho^{k}}=D_\rho(v_1,G)\cup \dots \cup
D_\rho(v_k,G)$. As the removal of  the vertices of $C$ from $G^*$
leaves no connected component in $G^*[V\setminus C]$ with more than
$|V|/2$ vertices and since $G^*$ is a supergraph of $G$, clearly,
the removal of  the vertices of $\mathbf{D_\rho^{k}}$ from $G$
leaves no connected component in $G[V \setminus
\mathbf{D_\rho^{k}}]$  with more than $|V|/2$ vertices. \qed
\end{proof}

Again, like in Section \ref{sec:k=1}, we do not assume that a
tree-decomposition $T(G)$ of $k$-breadth $\rho$ is given for $G$ as
part of the input. Our method does not need to know $T(G)$. For a
given graph $G$, integers $k\geq 1$ and $\rho\geq 0$, even checking
whether $G$ has a tree-decomposition of $k$-breadth $\rho$ is a hard
problem (as $\tb_{k}(G)=0$ if and only if $\tw(G)\leq k-1$).

Let $G$ be an arbitrary connected $n$-vertex $m$-edge graph. In
\cite{APPROX2011}, an algorithm was described which, given $G$ and
its arbitrary fixed vertex $v$, finds in $O(m)$ time a balanced disk
separator $D_r(v,G)$ of $G$ centered at $v$ and with minimum $r$. We
can use this algorithm as a subroutine to find for $G$ in $O(n^{k}
m)$ time a balanced $\mathbf{D_r^k}$-separator with minimum $r$.
Given arbitrary $k$ vertices $v_1,v_2, \dots , v_k$ of $G$, we can
add a new dummy vertex $x$ to $G$ and make it adjacent to only
$v_1,v_2, \dots , v_k$ in $G$. Denote the resulting graph by $G+x$.
Then, a balanced disk separator $D_{r+1}(x,G+x)$ of $G+x$ with
minimum $r+1$ gives a balanced separator of $G$ of the form
$D_r(v_1,G)\cup \dots \cup D_r(v_k,G)$ (for particular disk centers
$v_1,v_2, \dots , v_k$) with minimum $r$. Iterating over all $k$
vertices of $G$, we can find a balanced $\mathbf{D_{r}^k}$-separator
of $G$ with the smallest (absolute minimum) radius $r$. Thus, we
have the following result.

\begin{proposition}\label{prop:find-sep-tbk}
Let $k$ be a positive integer (assumed to be small). For an
arbitrary graph $G$ with $n\geq k$ vertices and $m$ edges, a
balanced $\mathbf{D_r^k}$-separator with the smallest radius $r$ can
be found in $O(n^{k} m)$ time.
\end{proposition}

\subsection{Decomposition of a graph with bounded
$k$-tree-breadth}\label{subsec:decomp-k}

Let $G=(V,E)$ be an arbitrary connected graph with $n$ vertices and
$m$ edges and with a balanced $\mathbf{D_r^k}$-separator, where
$\mathbf{D_r^k}=\bigcup_{j=1}^k D_r(v_j,G)$.
Note that some disks in $\{D_r(v_1,G),\dots,D_r(v_k,G)\}$ may
overlap. In what follows, we will partition
$\mathbf{D_r^k}=\bigcup_{j=1}^k D_r(v_j,G)$ into $k$ sets
$D_1,\dots, D_k$ such that no two of them intersect and each $D_j$,
$j=1,\dots,k$, contains at least one vertex $v_j$ and induces a
connected subgraph of $G[D_r(v_j,G)]$. Create a graph $G+s$ by
adding a new dummy vertex $s$ to $G$ and making it adjacent to only
$v_1,v_2,\dots, v_k$ in $G$. Let $T$ be a BFS-tree of $G+s$ started
at  vertex $s$ and $T'$ be a subtree of $T$ formed by vertices
$\{v\in V(G+s)| d_T(s,v)\leq r+1\}$ and rooted at $s$. Let also
$T(v_1),\dots,T(v_k)$ be the subtrees of $T'\setminus\{s\}$ rooted
at $v_1,\dots,v_k$, respectively. Clearly, each $T(v_j)$, $j=1,
\dots,k$, is a subtree (not necessarily spanning) of $G[D_r(v_j,G)]$
and $\mathbf{D_r^k}=\bigcup_{j=1}^k V(T(v_j))$. Set now
$D_j:=V(T(v_j))$, $j=1, \dots,k$.

Let $G_1,G_2,\dots ,G_q $ be the connected components of
$G[V\setminus \mathbf{D_r^k}]$. Denote by $S_i^j=\{v\in V(G_i)|
d_G(v,D_j)=1\}$, $i=1,\dots,q$, $j=1,\dots,k$, the neighborhood of
$D_j$ in $G_i$. Also, let $G_i^+$ be the graph obtained from
component $G_i$ by adding one {\em meta vertex} $c_i^j$  for each
disk $D_r(v_j,G)$ (a {\em representative of} $D_r(v_j,G)$), $j=1,
\dots k$,  and making it adjacent to all vertices of $S_i^j$, i.e.,
for a vertex $x\in V(G_i),\; c_i^jx \in E(G_i^+)$ if and only if
there is a vertex $x_D \in D_j\subseteq D_r(v_j,G)$ with $xx_D \in
E(G)$. If $S_i^j$ is empty for some $j$, then vertex $c_i^j$ is not
added to $G_i^+$. Also, add an edge between any two representatives
$c_i^j$ and $c_i^l$ if vertices $v_j$ and $v_l$ are connected in
$G[V\setminus V(G_i)]$.
See Fig.~\ref{fig:Drk-sep} for an illustration.

Given an $n$-vertex $m$-edge graph $G$ and its balanced
$\mathbf{D_r^k}$-separator, the graphs $G_1^+,\dots, G_q^+$ can be
constructed in total time $O(k q m)$. Furthermore, the total number
of edges in graphs $G_1^+,\dots, G_q^+$ does not exceed $ m+q k^2$,
and the total number of vertices in those graphs does not exceed the
number of vertices in $G[V\setminus \mathbf{D_r^k} ]$ plus $qk$.

\begin{figure}[htb]
\begin{center} \vspace*{-5mm}
 \begin{minipage}[b]{15cm}
    \begin{center} \hspace*{10mm}
\includegraphics[height=5.cm]{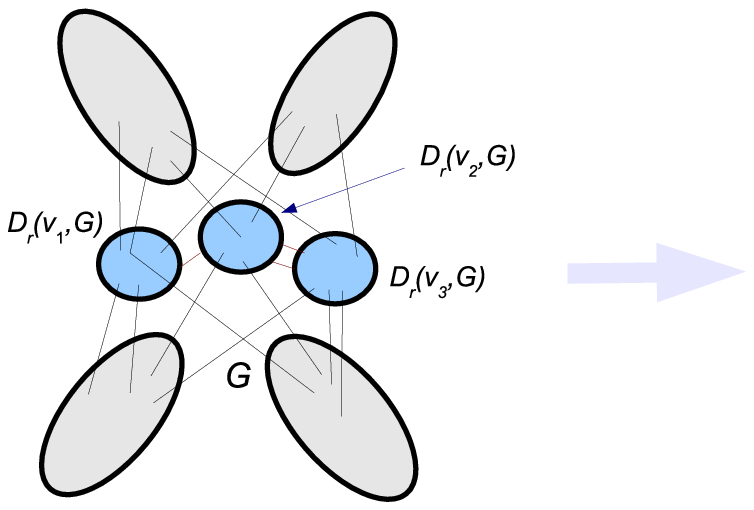}\hspace*{-12mm}
\includegraphics[height=5.cm]{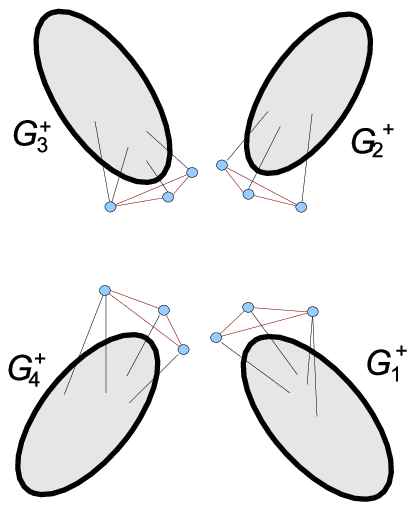}
    \end{center} \vspace*{-8mm}
    \caption{\label{fig:Drk-sep} A graph $G$ with a balanced $\mathbf{D_r^3}$-separator  and the corresponding graphs $G_1^+,\dots, G_4^+$ obtained from $G$.
  Each $G_i^+$ has three  meta vertices representing the three disks. } %
 \end{minipage}
\end{center}
\vspace*{-6mm}
\end{figure}

Note that $G_i^+$ is a minor of $G$ and can be obtained from $G$ by
a sequence of edge contractions in the following way. First contract
all edges (in any order) that are incident to $V(G_{i'})$, for all
$i'=1,\dots,q$, $i'\neq i$. Then, for each $j=1, \dots,k$, contract
(all edges of) connected subgraph $G[D_j]$ of $G$ to get meta vertex
$c_i^j$ representing the disk $D_r(v_j,G)$ in $G_i^+$.

Let again $G_{/e}$ be the graph obtained from $G$ by contracting
edge $e$. We have the following analog of Lemma
\ref{lem:G+hereditary}.

\begin{lemma}\label{lem:hereditary-tbk}
For any graph $G$ and its edge $e$, $\tb_k(G)\leq \rho$ implies
$\tb_k(G_{/e}) \leq \rho$. Consequently, for any graph $G$ with
$\tb_k(G)\leq \rho$, $\tb_k(G_i^+) \leq \rho$ holds for $i=1,
\dots,q$.
\end{lemma}

\begin{proof}
Let ${T}(G)=(\{X_i|i\in I\},{T}=(I,F))$ be a tree-decomposition of
$G$ with $k$-breadth $\rho$. Let $e=xy$ be an arbitrary edge of $G$.
We can obtain a tree-decomposition ${T}(G_{/e})$ of the graph
$G_{/e}$ by replacing in each bag $X_i$, $i\in I$, vertices $x$ and
$y$ with a new vertex $x'$ representing them (if some bag $A$
contained both $x$ and $y$, only one copy of $x'$ is kept).
Evidently, the first and the second conditions of
tree-decompositions are fulfilled for $T(G_{/e})$. Furthermore, the
topology  (the tree ${T}=(I,F)$) of the tree-decomposition did not
change. Still, for any vertex $v\neq x'$ of $G_{/e}$, the bags of
$T(G_{/e})$ containing $v$ form a subtree in $T(G_{/e})$. Since
vertices $x$ and $y$ were adjacent in $G$, there was a bag $A$ of
$T(G)$ containing both those vertices. Hence, a subtree of
$T(G_{/e})$ formed by bags of $T(G_{/e})$ containing vertex $x'$ is
nothing else but the union of two subtrees (one for $x$ and one for
$y$) of $T(G)$ sharing at least one common bag $A$. Also,
contracting an edge can only reduce the distances in a graph. Hence,
still, for each bag $B$ of $T(G_{/e})$, there must exist
corresponding vertices $v_1,\dots,v_k$ in $G_{/e}$ with $B\subseteq
D_\rho(v_1,G_{/e})\cup \dots \cup D_\rho(v_k,G_{/e})$. Thus,
$\tb_k(G_{/e})\leq \rho$. Since $G_i^+$ can be obtained from $G$ by
a sequence of edge contractions, we also have $\tb_k(G_i^+) \leq
\rho$. \qed
\end{proof}


\subsection{Construction of a hierarchical tree}\label{subsec:H(G)}
Here we show how a hierarchical tree for a graph with bounded
$k$-tree-breadth is built.

Let $G=(V,E)$ be a connected $n$-vertex, $m$-edges graph  with
$\tb_k(G) \leq \rho$ and $n\geq k$. Lemma~\ref{disks_sep} guaranties
that $G$ has a balanced $\mathbf{D_r^k}$-separator with $r\leq
\rho$. Proposition \ref{prop:find-sep-tbk} says that such a balanced
$\mathbf{D_r^k}$-separator of $G$ can be found in $O(n^k m)$ time by
an algorithm that works directly on the graph $G$ and does not
require construction of a tree-decomposition of $G$ with $k$-breadth
$\leq \rho$. Using these and Lemma \ref{lem:hereditary-tbk}, we can
build a rooted \emph{hierarchical-tree} $\mathcal{H}(G)$ for $G$,
which is constructed as follows. If $G$ is a connected graph with at
most $2k+1$ vertices, then $\mathcal{H}(G)$ is one node tree with
root node $(V(G),G)$. It is known \cite{Domin98a} that any graph
with $p\geq 2$ vertices has a dominating set of size $\lfloor p/2
\rfloor$, i.e., all vertices of it can be covered by $\lfloor p/2
\rfloor$ disks of radius one. Hence, in our case, $G$ with at most
$2k+1$ vertices can be covered by $k$ disks of radius one each,
i.e., there are $k$ vertices $v_1,\dots,v_k$ such that
$V(G)=D_r(v_1,G)\cup \dots \cup D_r(v_k,G)$ for $r=1\leq \rho$. If
$G$ is a connected graph with more than $2k+1$ vertices,
find a balanced $\mathbf{D_r^k}$-separator of minimum radius $r$
in $O(n^k m)$ time and construct the corresponding graphs
$G_1^+,\dots,G_q^+$. For each graph $G_i^+,\; i\in \{1,\dots ,q\}$,
(by Lemma \ref{lem:hereditary-tbk}, $\tb_k(G_i^+) \leq \rho$)
construct a hierarchical tree $\mathcal{H}(G_i^+)$ recursively and
build $\mathcal{H}(G)$ by taking the pair $(\mathbf{D_r^k},G)$ to be
the root and connecting the root of each tree $\mathcal{H}(G_i^+)$
as a child of $(\mathbf{D_r^k},G)$.

The depth of this tree $\cH(G)$ is the smallest integer $p$ such
that
$$\frac{n}{2^p}+k(\frac{1}{2^{p-1}}+\dots +\frac{1}{2}+1)\leq 2k+1,$$
that is, the depth 
is at most $\log_2 n$. It is also not hard to see that, given a
graph $G$ with $n$ vertices and $m$ edges, a hierarchical tree
$\cH(G)$
can be constructed in $O((kn)^{k+2}\log^{k+1}n)$ total time. There
are at most $O(\log n)$ levels in  $\cH(G)$, and one needs to do at
most $O((n+k n \log n)^k(m+k^2 n \log n))\leq
O((kn)^{k+2}\log^{k}n)$ operations per level since the total number
of edges in the graphs of each level is at most $O(m+k^2 n \log n)$
and the total number of vertices in those graphs can not exceed
$O(n+k n \log n)$.

For nodes of $\cH(G)$, we use the same notations as in Section
\ref{sec:k=1}. For a node $Y$ of $\mathcal{H}(G)$, since
it is associated with a pair $(\mathbf{D_{r'}^k},G')$, where $r'\leq
\rho$, $G'$ is a minor of $G$ and $\mathbf{D_{r'}^k}=D_{r'}(v'_1,G')
\cup \dots \cup D_{r'}(v'_1,G')$,
it is convenient to denote $G'$ by $G(\downarrow Y)$,
$\{v'_1,\dots,v'_k\}$ by $c(Y)=\{c_{1}(Y),\dots,c_{k}(Y)\}$, $r'$ by
$r(Y)$, and $\mathbf{D_{r'}^k}$ by $Y$ itself. Thus,
$(\mathbf{D_{r'}^k},G')=(\bigcup_{l=1}^k
D_{r(Y)}(c_{l}(Y),G(\downarrow Y)),G(\downarrow Y))=(Y, G(\downarrow
Y))$ in these notations, and we identify node $Y$ of
$\mathcal{H}(G)$ with the set $\bigcup_{l=1}^k
D_{r(Y)}(c_{l}(Y),G(\downarrow Y))$ and associate with this node
also the graph $G(\downarrow Y)$.
If now $(Y^0,Y^1,\dots, Y^h)$ is the path of $\cH(G)$ connecting the
root $Y^0$ of $\cH(G)$ with a node $Y^h$, then the vertex set of the
graph $G(\downarrow Y^h)$ consists of some (original) vertices of
$G$ plus at most $kh$ meta vertices representing the disks
$D_{r(Y)}(c_{1}(Y^i),G(\downarrow Y^i)),\dots,
D_{r(Y)}(c_{k}(Y^i),G(\downarrow Y^i))$ of $Y^i$, $i=0,1,\dots,h-1$.
Note also that each (original) vertex of $G$ belongs to exactly one
node of $\mathcal{H}(G)$.

\subsection{Construction of collective additive tree spanners}
Let $G=(V,E)$ be a connected  $n$-vertex, $m$-edge graph and assume
that $\tb_k(G)\leq\rho$ and $n\geq k$. Let $\cH(G)$ be a
hierarchical tree of $G$. Consider an arbitrary node  $Y^h$ of
$\cH(G)$, and let $(Y^0,Y^1,\dots, Y^h)$ be the path of $\cH(G)$
connecting the root $Y^0$ of $\cH(G)$ with $Y^h$. Let
$\widehat{G}(\dY^j)$ be the graph obtained from $G(\dY^j)$ by
removing all its meta vertices (note that $\widehat{G}(\dY^j)$ may
be disconnected and that all meta vertices of $G(\dY^j)$ come from
previous levels of $\cH(G)$). We have the following analog of Lemma
\ref{lem:closest-center}.

\begin{lemma}\label{lem:closest-center-anyk}
For any vertex $z$ from $Y^h\cap V(G)$  there exists an index $i\in
\{0,1,\dots,h\}$ such that the vertices $z$ and $c_{l}(Y^i)$, for
some $l \in \{1,\dots,k\}$ can be connected in the graph
$\widehat{G}(\downarrow Y^i)$ by a path of length at most
$\rho(h+1)$. In particular, $d_G(z,c_{l}(Y^i))\leq \rho(h+1)$ holds.
\end{lemma}

\begin{proof}
The proof is similar to the proof of Lemma \ref{lem:closest-center}
of Section \ref{sec:k=1}. Set $G_h:=G(\downarrow Y^h)$ and
$c:=c_{l}(Y^h)$, where $z \in D_l\subseteq
D_{r(Y^h)}(c_{l}(Y^h),G_h)$ (for the definition of set $D_l$ see the
first paragraph of Subsection \ref{subsec:decomp-k}).   Let
$SP^{G_h}_{c,z}$ be a shortest path of $G_h$ connecting vertices $c$
and $z$. We know that this path has at most $r(Y^h)\leq \rho$ edges.
If $SP^{G_h}_{c,z}$ does not contain any meta vertices, then this
path is a path of $\widehat{G}(\downarrow Y^h)$ and of $G$ and
therefore $d_G(c,z)\leq \rho$ holds.

Assume now that $SP^{G_h}_{c,z}$ does contain meta vertices and let
$\mu'$ be the closest to $z$ meta vertex in $SP_{c,z}^{G_h}$
(consult with Fig.~\ref{fig:cc}). Let
$SP^{G_h}_{c,z}=(c,\dots,a',\mu',b',\dots,z)$. By construction of
$\cH(G)$, meta vertex $\mu'$ was created at some earlier recursive
step to represent one disk of $Y^{i'}$ of graph
$G_{i'}:=G(\downarrow Y^{i'})$ for some $i'\in \{0,\dots,h-1\}$.
Hence, there is a path $P^{G_{i'}}_{c',z}=(c',\dots,b',\dots,z)$ of
length at most $2\rho$ in $G_{i'}$ with $c':=c_{l'}(Y^{i'})$ for
some $l' \in \{1,\dots,k\}$. Again, if $P^{G_{i'}}_{c',z}$ does not
contain any meta vertices, then this path is a path of
$\widehat{G}(\downarrow Y^{i'})$ and of $G$ and therefore
$d_G(c',z)\leq 2\rho$ holds. If $P^{G_{i'}}_{c',z}$ does contain
meta vertices then again, ``unfolding" a meta vertex $\mu''$ of
$P^{G_{i'}}_{c',z}$ closest to $z$, we obtain a path
$P^{G_{i''}}_{c'',z}$ of length at most $3\rho$ in
$G_{i''}:=G(\downarrow Y^{i''})$ with $c'':=c_{l''}(Y^{i''})$ for
some $i''\in \{0,\dots,i'-1\}$ and $l'' \in \{1,\dots,k\}$.

We continue ``unfolding" this way meta vertices closest to $z$.
Eventually, after at most $h$ steps, we will arrive at the situation
when,  for some index $i^*\in \{0,1,\dots,h\}$, a path of length at
most $\rho(h+1)$ will connect vertices $z$ and $c_{l^*}(Y^{i^*})$,
for some $l^* \in \{1,\dots ,k\}$, in the graph
$\widehat{G}(\downarrow Y^{i^*})$.\qed
\commentout{ The proof is similar to the proof of Lemma
\ref{lem:closest-center} of Section \ref{$sec:k=1$}. Each unfolding
operation of the path from the closest meta vertex to $z$ might
introduce only meta vertices of lower levels nodes of $\cH(G)$ and
since we pick the closets one to $z$ for next step of unfolding.
Then, we actually unfold a meta vertex from a lower level each time
and no two meta vertices from the same level (or node of $\cH(G)$)
would be picked up for unfolding. Thus we could have at most $h$
unfolding steps.
}
\end{proof}

Let $B_1^i,\dots, B_{p_i}^i$ be the nodes at depth $i$ of the tree
$\cH(G)$. Assume $B_j^i=\bigcup_{l=1}^{k} D_r(c_j^i(l),G(\downarrow
B_j^i))$, where $r:=r(B_j^i)$. Denote $k$ central vertices of
$B_j^i$ by $C_j^i=\{c_j^i(1),c_j^i(2),\dots,$ $c_j^i(k)\}$. For each
node $B^i_j$, consider its (central) vertex $c_j^i(l)$ $(l
\in\{1,\dots,k\})$. If $c_j^i(l)$ is an original vertex of $G$ (not
a meta vertex created during the construction of $\cH(G)$), then
define a  connected graph $G_j^i(l)$
obtained from $G(\downarrow B_j^i)$ by removing all its meta
vertices.  
If removal of those meta vertices produced few connected components,
choose as $G_j^i(l)$ that component which contains the vertex
$c_j^i(l)$. Denote by $T_j^i(l)$ a BFS--tree of graph $G_j^i(l)$
rooted at vertex $c_j^i(l)$ of
$B_j^i$. 

The trees $T_j^i(l)$ ($i=0,1,\dots,depth(\cH(G))$,
$j=1,2,\dots,p_i$, $l=1,2,\dots,k)$, obtained this way, are called
{\em local subtrees} of $G$.  Clearly, the construction of these
local subtrees can be incorporated into the procedure of
constructing hierarchical tree $\cH(G)$ of $G$ and will not increase
the overall $O((kn)^{k+2}\log^{k+1}n)$ run-time (see Subsection
\ref{subsec:H(G)}).

Since Lemma  \ref{lm:decomp} and Lemma \ref{lm:isom} hold for $G$,
similarly to the proof of Lemma \ref{lm:distT}, one can prove its
analog for graphs with bounded $k$-tree-breadth.

\begin{lemma} \label{lm:distT2}
For any two vertices $x,y \in V(G)$, there exists a local subtree
$T$ such that $d_{T}(x,y)\leq d_{G}(x,y)+2\rho (1+\log_2 n)$.
\end{lemma}

This lemma implies the following two results. Let $G$ be a graph
with $n$ vertices and $m$ edges having $\tb_k(G)\leq \rho$. Let also
$\cH(G)$ be its hierarchical tree and $\mathcal{LT}(G)$ be the
family of all its local subtrees (defined above). Consider a graph
$H$ obtained by taking the union of all local subtrees of $G$ (by
putting all of them together).
Clearly, $H$ is a spanning subgraph of $G$, constructible in
polynomial time for every fixed $k$. %
We have  $d_H(x,y)\leq d_G(x,y)+2\rho(1+\log_2 n)$ for any two
vertices $x$ and $y$ of $G$. Also, since for every level $i$
($i=0,1,\dots,depth(\cH(G))$) of hierarchical tree $\cH(G)$, the
corresponding local subtrees $T^i_1(l),\dots, T^i_{p_i}(l)$ for each
fixed index $l \in\{1,\dots,k\}$ are pairwise vertex-disjoint, their
union has at most $n-1$ edges. Therefore, $H$ cannot have more than
$k(n-1)(1+\log_2 n)$ edges in total. Thus, we have the following
result.

\begin{theorem} \label{tm:spanner-anyk}
Every graph $G$ with $n$ vertices and $\tb_{k}(G)\leq \rho$ admits
an additive $(2\rho(1+\log_2 n))$--spanner with at most $O(k n\log
n)$ edges constructible in polynomial time for every fixed $k$.
\end{theorem}

For a node $B_j^i$ of $\cH(G)$, let
$\cT_j^i=\{T_j^i(1),\dots,T_j^i(k)\}$ be the set of its local
subtrees. Instead of taking the union of all local subtrees of $G$,
one can fix $i$ ($i\in\{0,1,\dots,depth(\cH(G))\}$) and fix $l \in
\{1,\dots,k\}$ and consider separately the union of only local
subtrees $T^i_1(l),\dots, T^i_{p_i}(l)$, corresponding to the
$l${th} subtrees of level $i$ of the hierarchical tree $\cH(G)$, and
then extend in linear $O(m)$ time that forest to a spanning tree
$T^i(l)$ of $G$ (using, for example, a variant of the Kruskal's
Spanning Tree algorithm for the unweighted graphs). We call this
tree $T^i(l)$ the $l$th {\em spanning tree of $G$ corresponding to
the level $i$ of the hierarchical tree $\cH(G)$}. In this way we can
obtain at most $k(1+\log_2 n)$ spanning trees for $G$, $k$ trees for
each level $i$ of $\cH(G)$. Denote the collection of those spanning
trees by $\cT(G)$. Thus, we deduce the following theorem.

\begin{theorem} \label{tm:system-anyk}
Every graph $G$ with $n$ vertices and $\tb_k(G)\leq \rho$ admits a
system $\cT(G)$ of at most $k(1+\log_2 n)$ collective additive tree
$(2\rho(1+\log_2 n))$-spanners constructible in polynomial time for
every fixed $k$.
\end{theorem}


\section{Additive Spanners for Graphs Admitting (Multiplicative)
$t$--Spanners of Bounded Tree-width.}\label{sec:conseq}
In this section, we show that if a graph $G$ admits a
(multiplicative) $t$-spanner $H$ with $\tw(H)=k-1$ then its
$k$-tree-breadth is at most $\lceil{t/2}\rceil$. As a consequence,
we obtain that, for every fixed $k$, there is a polynomial time
algorithm that, given an $n$-vertex graph $G$ admitting a
(multiplicative) $t$-spanner with tree-width at most $k-1$,
constructs a system of at most $k (1+\log_2 n)$ collective additive
tree $O(t\log n)$-spanners of $G$.

\subsection{$k$-Tree-breadth of a graph admitting  a $t$-spanner of bounded tree-width}

Let $H$ be a graph with tree-width $k-1$, and let $T(H)=(\{X_i|i\in
I\},T=(I,F))$ be its tree-decomposition of width $k-1$. For an
integer $r\geq 0$, denote by $X^{(r)}_i$, $i\in I$, the set
$D_r(X_i,H):=\bigcup_{x\in X_i} D_r(x,H)$. Clearly, $X^{(0)}_i=X_i$
for every $i\in I$. The following important lemma holds.

\begin{lemma}\label{lem:decomp-exp} For every integer $r\geq 0$,
 $T^{(r)}(H):=(\{X^{(r)}_i|i\in I\},T=(I,F))$ is a tree-decompo\-sition of $H$ with $k$-breadth $\leq
 r$.
\end{lemma}

\begin{proof}
It is enough to show that the third condition of tree-decompositions
(see Subsection \ref{our-res}) is fulfilled for $T^{(r)}(H)$. That
is, for all $i,j,k \in I$, if $j$ is on the path from $i$ to $k$ in
$T$, then $X^{(r)}_i \bigcap X^{(r)}_k\subseteq X^{(r)}_j$.
We know that $X_i \bigcap X_k\subseteq X_j$ holds and
need to show that for every vertex $v$ of $H$, $d_H(v,X_i)\leq r$
and $d_H(v,X_k)\leq r$ imply $d_H(v,X_j)\leq r$. Assume, by way of
contradiction,  that for some integer $r>0$ and for some vertex $v$
of $H$, $d_H(v,X_j)> r$ while $d_H(v,X_i)\leq r$  and
$d_H(v,X_k)\leq r$.

Consider the original tree-decomposition $T(H)$. It is known
\cite{Diestel00} that if $ab$ ($a,b\in I$) is an edge of the tree
$T=(I,F)$ of tree-decomposition $T(H)$, and $T_a$, $T_b$ are the
subtrees of $T$ obtained after removing edge $ab$ from $T$, then
$S=X_a\cap X_b$ separates in $H$ vertices belonging to bags of $T_a$
but not to $S$ from vertices belonging to bags of $T_b$ but not to
$S$. We will use
this nice separation property. 

Let $T\setminus \{j\}$ be the forest obtained from $T$ by removing
node $j$, and let $T(i)$ and $T(k)$ be the trees from this forest
containing nodes $i$ and $k$, respectively. Clearly, $T(i)$ and
$T(k)$ are disjoint. The above separation property and inequalities
$d_H(v,X_i)\leq r < d_H(v,X_j)$ ensure that the vertex $v$ belongs
to a node (a bag) of $T(i)$ ($X_j$ cannot separate in $H$ vertex $v$
from a vertex $x_i$ of $X_i$ with $d_H(v,X_i)=d_H(v,x_i)$ since
otherwise $d_H(v,X_i)> d_H(v,X_j)$ will hold). Similarly,
inequalities $d_H(v,X_k)\leq r < d_H(v,X_j)$ and the above
separation property guaranty that the vertex $v$ belongs to a node
of $T(k)$. But then, the third condition of tree-decompositions says
that $v$ must also belong to the bag $X_j$ of $T(H)$. The latter,
however, is in a contradiction with the assumption that $d_H(v,X_j)>
r\geq 0$. \qed
\end{proof}


Now we can prove the main lemma of this section.

\begin{lemma}\label{lem:spanner-ktb}
 If a graph $G$ admits a $t$-spanner with tree-width $k-1$, then $\tb_{k}(G) \leq \lceil t/2 \rceil$.
\end{lemma}

\begin{proof}
Let $H$ be a $t$-spanner of $G$ with $\tw(G)=k-1$ and
$T(H)=(\{X_i|i\in I\},T=(I,F))$ be a tree-decomposition of $H$ of
width $k-1$. We claim that $T(G):=T^{(\lceil t/2
\rceil)}(H):=(\{X^{(\lceil t/2 \rceil)}_i|i\in I\},T=(I,F))$ is a
tree-decomposition of $G$ with $k$-breadth $\leq \lceil t/2 \rceil$.

By Lemma \ref{lem:decomp-exp}, $T^{(\lceil t/2
\rceil)}(H):=(\{X^{(\lceil t/2 \rceil)}_i|i\in I\},T=(I,F))$ is a
tree-decomposition of $H$ with $k$-breadth $\leq \lceil t/2 \rceil$.
Hence, the first and the third conditions of tree-decompositions
hold for $T(G)$.  For every pair $u,v$ of vertices of $G$,
$d_G(u,v)\leq d_H(u,v)$. Therefore, every disk $D_{\lceil t/2
\rceil}(x,H)$ of $H$ is contained in a disk $D_{\lceil t/2
\rceil}(x,G)$ of $G$.  This implies that  every  bag of $T(G)$ is
covered by at most $k$ disks of $G$ of radius at most $\lceil t/2
\rceil$ each, i.e.,
$$X^{(\lceil t/2 \rceil)}_i=D_{\lceil t/2
\rceil}(X_i,H)=\bigcup_{x\in X_i} D_{\lceil t/2
\rceil}(x,H)\subseteq \bigcup_{x\in X_i} D_{\lceil t/2
\rceil}(x,G).$$

We need only to show additionally that each edge $uv$ of $G$ belongs
to some bag of $T(G)$. Since $H$ is a $t$-spanner of $G$,
$d_H(u,v)\leq t$ holds. Let $x$ be a middle vertex of a shortest
path connecting $u$ and $v$ in $H$. Then, both $u$ and $v$ belong to
the disk $D_{\lceil t/2 \rceil}(x,H)$. Let $X_i$ be a bag of $T(H)$
containing vertex $x$. Then, both $u$ and $v$ are contained in
$X^{(\lceil t/2 \rceil)}_i$, a bag of $T(G)$.  \qed
 \end{proof}

\subsection{Consequences}
Now we give two implications of the above results for the class of
graphs admitting (multiplicative) $t$--spanners  with tree-width
$k-1$. They are direct consequences of Lemma ~\ref{lem:spanner-ktb},
Theorem \ref{tm:spanner-anyk} and Theorem ~\ref{tm:system-anyk}.

\begin{theorem}
Let $G$ be a graph with $n$ vertices and $m$ edges having a
(multiplicative) $t$--spanner with tree-width $k-1$. Then, $G$
admits an additive $(2\lceil{t/2}\rceil(1+\log_2 n))$--spanner with
at most $O(k n\log n)$ edges constructible in polynomial time for
every fixed $k$. 
\end{theorem}

\begin{theorem}
Let $G$ be a graph with $n$ vertices and $m$ edges having a
(multiplicative) $t$--spanner with tree-width $k-1$. Then,  $G$
admits a system $\cT(G)$ of at most $k(1+\log_2 n)$ collective
additive tree $(2\lceil{t/2}\rceil(1+\log_2 n))$--spanners
constructible in polynomial time for
every fixed $k$. 
\end{theorem}

\section{Concluding Remarks and Open Problems} \label{sec:concl}
Using Robertson-Seymour's tree-decomposition of graphs, we described
a necessary condition for a graph to have a multiplicative
$t$-spanner of tree-width $k$ (in particular, to have a
multiplicative tree $t$-spanner, when $k=1$). As we have mentioned
earlier, this necessary condition is far from being sufficient. The
following interesting problem remains open.
\begin{itemize}\vspace*{-2mm}
  \item Does there exist a clean ``if and only if" condition under which a graph admits
  a multiplicative (or, additive) $t$-spanner of tree-width $k$ (in particular, admits
  a multiplicative (or, additive) tree $t$-spanner ($k=1$ case))?
\end{itemize}\vspace*{-2mm}
That necessary condition was very useful in demonstrating that, for
every fixed $k$, there is a polynomial time algorithm that, given an
$n$-vertex graph $G$ admitting a multiplicative $t$-spanner with
tree-width $k$, constructs a system of at most $(k+1)(1+ \log_2 n)$
collective additive tree $O(t\log n)$-spanners of $G$. In
particular, when $k=1$, we showed that there is a polynomial time
algorithm that, given an $n$-vertex graph $G$ admitting a
multiplicative tree $t$-spanner, constructs a system of at most
$\log_2 n$ collective additive tree $O(t\log n)$-spanners of $G$.
Can these results be improved?
\begin{itemize}\vspace*{-2mm}
  \item Does a polynomial time algorithm exist that, given an $n$-vertex
graph $G$ admitting a multiplicative tree $t$-spanner, constructs a
system of $O(1)$ collective additive tree $O(t)$-spanners of $G$?
  \item Does a polynomial time algorithm exist that, given an $n$-vertex
graph $G$ admitting a multiplicative $t$-spanner with tree-width
$k$, constructs a system of $O(k)$ collective additive tree
$O(t)$-spanners of $G$?
\end{itemize}\vspace*{-2mm}
As we have mentioned earlier, an interesting particular question
whether a multiplicative tree spanner can be turned into an (one)
additive tree spanner with a slight increase in the stretch is
(negatively) settled already in \cite{EmekP04}.


Two more interesting challenging questions we leave for future
investigation.
\begin{itemize}\vspace*{-2mm}
  \item Is there any polynomial time algorithm which, given a graph
admitting a system of at most $\mu$ collective tree $t$-spanners,
constructs a system of at most $\alpha(\mu,n)$ collective tree
$\beta(t,n)$-spanners, where $\alpha(\mu,n)$ is $O(\mu)$ $($or
$O(\mu\log n))$ and $\beta(t,n)$ is $O(t)$ $($or $O(t\log n))$?
   \item Is there a polynomial time algorithm that, for every unweighted
graph $G$ admitting a $t$-spanner of tree-width $k$, constructs a
$(O(k\log n)t)$-spanner with tree-width at most $k$?
\end{itemize}\vspace*{-2mm}

\commentout{

\definef{CHECK: Find disk separator  with minimum $r$. If $r$ is big
then $\rho$ is big, too ($r\leq \rho$). Hence answer is NO. If $r$
is small then continue decomposing and building.   ... }

}

\end{document}